\documentclass[leqno, 10pt]{amsart}

\usepackage{geometry}
\usepackage{enumerate}
\newcommand\bx{{\mathbf x}}
\newcommand\by{{\mathbf y}}
\newcommand\bz{{\mathbf z}}

\newcommand\bw{{\mathbf w}}

\newcommand\LL{{\mathbb L}}

\newcommand\PP{{\mathbb P}}

\newcommand\RR{{\mathbb R}}
\newcommand\TT{{\mathbb T}}
\newcommand\ZZ{{\mathbb Z}}
\newcommand\CC{{\mathbb C}}

\newcommand\ve{\varepsilon}

\newcommand{\mc}[1]{{\mathcal #1}}
\newcommand{\mf}[1]{{\mathfrak #1}}

\usepackage{amssymb, amsmath}
\usepackage[none,bottom]{draftcopy}

\newtheorem{prop}{Proposition}
\newtheorem{theo}{Theorem}
\newtheorem{Lemma}{Lemma}

\title[]{Homogenization results for a linear dynamics in random Glauber type environment}

\author{C\'edric Bernardin}
\email{Cedric.Bernardin@umpa.ens-lyon.fr}
\address{%
Universit\'e de Lyon and CNRS, UMPA, UMR-CNRS 5669, ENS-Lyon,
46, all\'ee d'Italie, 69364 Lyon Cedex 07 - France.
}%

\date{\today}

\thanks{\textsc{Acknowledgements.} We thank J. Fritz and S. Olla for useful discussions on \cite{FFL}. This research has been partially supported by the French Ministry of Education through the ANR-10-BLAN 0108 (SHEPI) grant and the Minist\`ere des Affaires \'etrang\`eres et europ\'eennes through the PHC Balaton 19458UK grant}


\begin{document}
\maketitle

\begin{abstract}
We consider an energy conserving linear dynamics that we perturb by a Glauber dynamics with random site dependent intensity. We prove hydrodynamic limits for this non-reversible system in random media. The diffusion coefficient turns out to depend on the random field only by its statistics. The diffusion coefficient defined through the Green-Kubo formula is also studied and its convergence to some homogenized diffusion coefficient is proved.   \\

On consid\`ere un syst\`eme d'\'equations differentielles lin\'eaires coupl\'ees conservant une certaine \'energie et l'on perturbe ce syst\`eme par une dynamique de type Glauber dont l'intensit\'e varie al\'eatirement site par site. Nous prouvons les limites hydrodyanmiques pour ce syst\`eme non r\'eversible en milieu al\'eatoire. Le coefficient de diffusion d\'epend de l'al\'ea uniquement par sa loi. Nous \'etudions aussi le coefficient de diffusion d\'efini par la formule de Green-Kubo et montrons la convergence de celle-ci vers un coefficient de diffusion homog\'en\'eis\'e.

\end{abstract}

\section{Introduction}

The derivation of hydrodynamic limits for interacting particle diffusive systems in random environment has attracted a lot of interest in the last decade. One of the first paper to consider such question is probably \cite{Fr} where hydrodynamic behavior of a one-dimensional Ginzburg-Landau model in the presence of random conductivities is studied. In \cite{Q0}, a lattice gas with random rates is considered and a complete proof of hydrodynamic limits has been given in \cite{FM}, \cite{Q}. 
Other systems have been investigated such as exclusion processes and zero-range processes   (\cite{F1,F2,FJL,GJ,JL,N}). Interacting particle systems evolving in random media are in general  of non-gradient.  Roughly speaking the gradient condition means that the microscopic current associated to the conserved quantity is already of gradient form. Otherwise the general non-gradient techniques (\cite{KL}, \cite{V}) consists in establishing a microscopic fluctuation-dissipation equation which permits to replace the current by a gradient plus a fluctuation term. But, if the system evolves in a random medium, such a decomposition  does not hold microscopically because the fluctuations induced by the random medium are too large, and it is only in a mesoscopic scale that this fluctuation-dissipation equation makes sense (\cite{FM}, \cite{Q}).

In \cite{GJ, JL}, by extending some ideas of \cite{N}, a simpler approach is proposed. The idea is to introduce a functional transformation of the empirical measure, which turns the system into a gradient-model, in such a way that the transformed empirical measure is very close to the original empirical measure. The advantage of the method is that it avoids the heavy machinery of the non-gradient tools but is unfortunately restricted to specific models. Even if the techniques developed in \cite{FM}, \cite{Q} seem to be more robust than the precedent approach, it is not clear that in some situations, as in the situation considered here, they can be applied without a substantial modification. 

The interacting particle system we consider is the following. To a simple energy conserving linear dynamics, flips with site dependent rates are superposed. Fix a sequence $(\gamma_x)_x$ of positive numbers and denote by $(\eta (t))_{t \ge 0}$ the Markov process with state space ${\mathbb R}^{\mathbb Z}$ and generator given by
\begin{equation}
\label{eq:generator}
({\mathcal L} f )(\eta)= ({\mathcal A}f)(\eta) + ({\mathcal S} f)(\eta), \quad f:{\mathbb R}^{\ZZ} \to \RR
\end{equation}
where
$$(\mathcal A f )(\eta)=\sum_{x \in \ZZ} (\eta_{x+1}-\eta_{x-1})\partial_{\eta_x}f$$
and
$$({\mathcal S}f)(\eta)= \sum_{x \in \ZZ} \gamma_x \left[ f(\eta^x)-f(\eta)\right]$$
with $\eta^x$ the configuration obtained from $\eta$ by flipping $\eta_x$: $(\eta^x)_z = \eta_z$ if $z \ne x$, $(\eta^x)_x= -\eta_x$. This system conserves the energy $\sum_{x} e_x$, $e_x =\eta_x^2 /2$, and the product of centered Gaussian probability measures with variance $T>0$ are invariant for the dynamics. 

Let $(\gamma_x)_x$ be a sequence satisfying (\ref{eq:gounds}) and (\ref{eq:limg}). For example, the sequence $(\gamma_x)_x$ is a realization of i.i.d. positive bounded below and above random variables with positive finite mean. We show (cf. Theorem \ref{th:hl}) that, starting from a local equilibrium state with temperature profile $T_0 =1/\beta_0$, the system evolves in a diffusive time scale following a temperature profile $T$, which is a solution of the heat equation
\begin{equation}
\label{eq:chaleur}
\begin{cases}
\partial_t T ={\bar \gamma}^{-1} \Delta T\\
T (0,\cdot) = \beta_0^{-1} (\cdot)
\end{cases}
\end{equation}
where $\bar \gamma$ is the average of the flip rates $\gamma_x$ defined by (\ref{eq:limg}).

One of the main interest of the model is its non-reversibility. To the best of our knowledge, it is the first time that hydrodynamic limits are established for a non-reversible interacting particle system evolving in a random medium. In fact, our first motivation was to work with a simplified version of the energy conserving model of heat conduction with random masses (\cite{B2}) and we think that some of the methods developed in this paper could be useful to study this model.       

The derivation of the hydrodynamic limits presents three difficulties: the first is that the system is non-gradient. The second one is that it is non-reversible and that the symmetric part ${\mathcal S}$ of the generator is very degenerate and  gives only few pieces of information on the ergodic properties of the system. The third difficulty is more technical. The state space is non-compact and the control of high energies is non-trivial. The first problem is solved by using the "corrected empirical measure" method introduced in \cite{GJ}, \cite{JL} and some special features of the model. For the second one, we apply in this context some deep ideas introduced in \cite{FFL} (see also \cite{OVY}). The third problem is solved by observing that the set of convex combinations of Gaussian measures is preserved by the dynamics. The control of large energies is then reduced to the control of large covariances. 

In the perspective to study heat conduction models with random masses our main interest lies in the properties of the diffusion coefficient (given here by $1/{\bar \gamma}$).  

The diffusion coefficient is also often expressed by the Green-Kubo formula, which is nothing but the space-time variance of the current at equilibrium. The Green-Kubo expression is only formal in the sense that a double limit (in space and time) has to be taken. For reversible systems, the existence is not difficult to establish. But for non-reversible systems even the convergence of the formula is challenging (\cite{Oc}). Let us remark that a priori the Green-Kubo formula depends on the particular realization of the disorder.

If we let aside the existence problem, widely accepted heuristic arguments predict the equality between the diffusion coefficient defined through hydrodynamics and the diffusion coefficient defined by the Green-Kubo formula.

The second main theorem of our paper shows that the homogenization effect also occurs for the Green-Kubo formula (see Theorem \ref{th:homo}): for almost every realization of the disorder, the Green-Kubo formula exists and is independent of the disorder. Unfortunately we did not succeed to prove that the value of the Green-Kubo formula is $1/{\bar \gamma}$.

The paper is organized as follows. In section \ref{sec:model} we define the system. The proof of hydrodynamic limits is given in section \ref{sec:hl}. The two main technical steps which are the derivation of a one block lemma and the control of high energies are postponed to sections \ref{sec:one} and \ref{sec:mom-bounds}. The study of the Green-Kubo formula is the content of the last section.

\section{The model}
\label{sec:model}

For any $\alpha>0$,  let $\Omega_\alpha$ be the set composed of configurations $\eta =(\eta_x)_{x \in \ZZ}$ such that $\| \eta\|_{\alpha} <+\infty$ where
\begin{equation*}
\| \eta \|_{\alpha}^2 = \sum_{x \in \ZZ} e^{- \alpha |x|} \eta_x^2
\end{equation*}
Let $\Omega=\cap_{\alpha>0} \Omega_\alpha$ be equipped with its natural product topology and its Borel $\sigma$-field. The set of Borel probability measures on $\Omega$ will be denoted by ${\mathcal P} (\Omega)$. We also introduce the set $C_0^k (\Omega)$, $k \ge 1$, composed of bounded local functions on $\Omega$ which are differentiable up to order $k$ with bounded partial derivatives.  

The time evolution of the process $(\eta (t))_{t \ge 0}$ can be defined as follows. Let $\{{\mc N}_x \, ; \, x \in \ZZ\}$ be a sequence of independent Poisson processes. We shall denote by $\gamma_x>0$ the intensity of ${\mc N}_x$. We assume there exist positive constants $\gamma_-$ and $\gamma_+$ such that
\begin{equation}
\label{eq:gounds}
\forall x \in \ZZ, \quad  \gamma_- \le \gamma_x \le \gamma_+
\end{equation}
For every realization of the random element ${\mc N}=({\mc N}_x)_{x \in \ZZ}$, consider the set of integral equations:
\begin{equation}
\label{eq:dyneq}
\eta_x (t) = (-1)^{{\mc N}_{x}(t)} \left( \eta_x (0) -\int_{0}^t (-1)^{{\mc N}_{x}(s)} (\eta_{x+1} (s) -\eta_{x-1} (s)) ds\right)
\end{equation}

For each initial condition $\sigma \in \Omega$ the equations (\ref{eq:dyneq}) can be solved by a classical iterative scheme. The solution $\eta(\cdot): =\eta(\cdot, \sigma)$ defines a strong Markov process with c\`adl\`ag trajectories. Moreover each path $\eta (\cdot, \sigma)$ is a continuous and differentiable function of the initial data $\sigma$ (\cite{DaPrato}, \cite{Fr}, \cite{FFL}). We define the corresponding semigroup $(P_t)_{t \ge 0}$ by  $(P_t f)(\sigma) ={\mathbb E}_{\mc N} (f(\eta (t, \sigma)))$ where ${\mathbb E}_{\mc N}$ denotes the expectation with respect to the Poisson clocks and $f$ is a bounded measurable function on $\Omega$. 

Since the state space is not compact Hille-Yosida theory can not be applied directly. Nevertheless, the differentiability with respect to initial conditions and stochastic calculus show that the Chapman-Kolmogorov equations 
\begin{equation*}
(P_t f )(\sigma) = f(\sigma) +\int_0^t ({\mathcal L} P_s f)(\sigma) ds, \quad f \in C_0^{1} (\Omega)
\end{equation*}
and
\begin{equation*}
(P_t f )(\sigma) = f(\sigma) +\int_0^t (P_s {\mathcal L} f)(\sigma) ds, \quad f \in C_0^{1} (\Omega)
\end{equation*}
are valid with ${\mathcal L}$ the formal generator defined by (\ref{eq:generator}).

The two Chapman-Kolmogorov equations permit to deduce that the probability measures $\nu \in {\mathcal P} (\Omega)$, which are invariant for $(\eta (t))_{t \ge 0}$, are characterized by the stationary Kolmogorov equation
\begin{equation*}
\int ({\mathcal L} f)(\eta) d\nu (\eta)=0 {\rm{ \; \; for\, all \;}}  f \in C_{0}^{1} (\Omega) 
\end{equation*}

In particular, every Gibbs measure $\mu_\beta$ with inverse temperature $\beta>0$ is a stationary probability measure. Observe that $\mu_\beta$ is nothing but the product of centered Gaussian probability measures on $\RR$ with variance $\beta^{-1}$. It is easy to show that $(P_t)_{t \ge 0}$ defines a strongly continuous contraction semigroup in ${\mathbb L}^2 (\mu_\beta)$  whose generator is a closed extension of ${\mathcal L}$.

In fact, the infinite volume dynamics is well approximated by the finite dimensional dynamics $\eta^{n} (t) = \left\{ \eta_x^n (t) \, ; \, x \in \ZZ\right\}$, $n \ge 2$. It is defined by the generator ${\mc L}_n ={\mc A}_n + {\mc S}_n$ where, for any function $f \in C_0^1 (\Omega)$, 
\begin{equation*}
{\mc A}_n f  =\sum_{x=-n+1}^{n-1} (\eta_{x+1} -{\eta_{x-1}}) \partial_{\eta_x} f \,  -\eta_{n-1} \partial_{\eta_n} f  + \eta_{-(n-1)} \partial_{\eta_{-n}} f
\end{equation*}
and
\begin{equation*}
({\mc S}_n f) (\eta)=\sum_{x=-n}^n \gamma_x \left[ f(\eta^x) -f(\eta) \right]
\end{equation*}
Observe that $\eta_x^n (t)$, $|x| >n$, do not change in time. Moreover,  the total energy $\sum_{x \in \ZZ} e_x$ is conserved by the finite dimensional dynamics. We denote by $(P_t^n)_{t \ge 0}$ the corresponding semigroup. Let us fix a positive time $T>0$, a parameter $\alpha>0$ and a function $\phi \in C_0^1 (\Omega)$. One can prove there exist constants $C_n:= C(n,\alpha,T, \phi)$, $n \ge 2$, such that
\begin{equation}
\label{eq:10}
\sup_{t \in [0,T]} \left| (P_t^n \phi)(\eta) - (P_t \phi) (\eta) \right | \le C_n \| \eta\|_{\alpha}^2
\end{equation}  
and
$$\lim_{n \to \infty} C_n =0$$ 

This approximation is only used in the proof of Lemma \ref{lem:lem9}. The proof of (\ref{eq:10}) in a similar context can be found in \cite{BOO}, chapter 2 (see also \cite{FFL}).

\section{Hydrodynamic limits}
\label{sec:hl}

For any function $u: {\ZZ} \to \RR$, the discrete gradient $\nabla u$ of $u$ is the function defined on $\ZZ$ by
$$\forall x \in \ZZ, \quad (\nabla u) (x) =u( x+1) -u(x)$$

The hydrodynamic limits are established in a diffusive scale. This means that we perform the time acceleration $t \to N^2 t$ and the space dilatation $x \to x/N$. In the rest of the paper, apart from section \ref{sec:GK},  the process $(\eta (t))_{t \ge 0}$ is the Markov process defined above with this time change. The corresponding generator is $N^{2} {\mathcal L}$.

The local conservation of energy $e_x =\eta_x^2 /2$ is expressed by the following microscopic continuity equation
$$e_x (t) -e_x (0) = -N^2 \int_0^t  (\nabla j_{x-1,x}) (\eta (s)) ds$$
where the current $j_{x,x+1}:=j_{x,x+1} (\eta)$ is defined by
$$j_{x,x+1}(\eta) =-\eta_x \eta_{x+1}$$

We denote by $C_0 (\RR)$ the space of continuous functions on $\RR$ with compact support and by $C^k_0 (\RR)$, $k \ge 1$, the space of compactly supported functions which are differentiable up to order $k$. Let ${\mc M}$ (resp. ${\mc M}^+$) be the space of Radon measures (resp. positive Radon measures) on $\RR$ endowed with the weak topology.  If $G \in C^2_{0} (\RR)$ and $m \in {\mc M}$ then $\langle m, G \rangle$ denotes the integral of $G$ with respect to $m$.

The empirical positive Radon measure $\pi_t^N \in {\mc M}^+$, associated to the process $e (t) := \{e_x (t)\, ; \, x \in \ZZ\}$, is defined by
\begin{equation*}
\pi_t^N (du) = \cfrac{1}{N} \sum_{x \in \ZZ} e_x (t) \,\delta_{x/N} (du)
\end{equation*}

Fix a strictly positive inverse temperature profile $\beta_0 : {\mathbb R} \to (0,+\infty)$ and a positive constant ${\bar \beta}$ such that 
\begin{equation}
\label{eq:supen}
\lim_{N \to \infty} \cfrac{1}{N^2} \sum_{x \in \ZZ}\left[  \cfrac{1}{\beta_0 (x/N)} -\cfrac{1}{{\bar \beta}}\right]^2 =0
\end{equation}
Denote by $\mu^N= \mu_{\beta_0 (\cdot)}^N \in {\mc P} (\Omega)$ the product probability measure  defined by
$$\mu_{\beta_0 (\cdot)}^N (d\eta)= \prod_{x \in \ZZ} g_{\beta_0(x/N)} (\eta_x) d\eta_x$$
where $g_\beta(u) du$ is the centered Gaussian probability measure on $\RR$ with variance $\beta^{-1}$. 

We assume that the initial state satisfies 
\begin{equation}
\label{eq:entinit}
H(\mu^N | \mu_{\bar \beta}) \le C_0 N
\end{equation}
for a positive constant $C_0$ independent of $N$. Here $H(\cdot| \cdot)$ is the relative entropy, which is defined, for two probability measures $P,Q \in {\mc P} (\Omega)$, by
\begin{equation}
\label{eq:vfh0}
H(P|Q) = \sup_{\phi} \left\{ \int \phi dP -\log \left( \int e^{\phi} dQ \right)  \right\}
\end{equation}   
with the supremum carried over all bounded measurable functions $\phi$ on $\Omega$. Let us recall the entropy inequality, which states that for every positive constant $a>0$ and every bounded measurable function $\phi$,
\begin{equation}
\label{eq:enti}
\int \phi \, dP \le a^{-1} \, \left\{ \log \left( \int e^{a \phi} dQ \right) + H(P | Q)  \right\}
\end{equation}

Fix a positive time $T>0$. The law of the process on the path space $D([0,T], \Omega)$, induced by the Markov process $(\eta(t))_{t \ge 0}$ starting from $\mu^N$, is denoted by ${\mathbb P}_{\mu^N}$. For any time $s\ge 0$, the probability measure on $\Omega$ given by the law of $\eta (s)$ is denoted by $\mu_s^N$. 

Since entropy is decreasing in time, (\ref{eq:entinit}) implies that
\begin{equation}
\label{eq:enttime}
\forall s \ge 0, \quad H(\mu^N_s | \mu_{\bar \beta}) \le C_0 N
\end{equation}

The conditions (\ref{eq:supen}) and (\ref{eq:entinit}) are introduced to get some moment bounds (see section \ref{sec:mom-bounds}). They are satisfied by any continuous function $\beta^{-1}_{0}$ going to ${\bar \beta}^{-1}$ at infinity sufficiently fast.

\begin{theo}
\label{th:hl}
Let $(\gamma_x)_{x\in \ZZ}$ be a sequence of positive numbers satisfying (\ref{eq:gounds}) and such that
\begin{equation}
\label{eq:limg}
\lim_{K \to \infty} \cfrac{1}{K} \sum_{x=1}^K \gamma_x ={\bar \gamma}, \quad \lim_{K \to \infty} \cfrac{1}{K} \sum_{x=-K}^0 \gamma_x ={\bar \gamma}
\end{equation}
for some ${\bar \gamma} \in (0,\infty)$.  Assume that the initial state $\mu^N=\mu_{\beta_0 (\cdot)}^N$ satisfies (\ref{eq:entinit}) and $\beta_0$ satisfies (\ref{eq:supen}).

Then, under ${\PP}_{\mu^N}$, $\pi_t^N$ converges in probability to $T_t /2$ where $T_t$ is the unique weak solution of (\ref{eq:chaleur}): For every $G \in C_0 (\RR)$, every $t>0$, and every $\delta>0$,
\begin{equation*}
\lim_{N \to \infty} {\mathbb P}_{\mu^N} \left[ \left| \langle \pi_t^N , G \rangle - \cfrac{1}{2}\, \langle T_t , G \rangle \right| \ge \delta \right]=0
\end{equation*}
\end{theo}

We follow the method of the ``corrected empirical measure'' introduced in \cite{GJ,JL}. Since the state space is not compact, technical adaptations are necessary. In particular, it is not given for free that the corrected empirical measure and the empirical measure have the same limit points for the weak convergence. It would be trivial if the state space was compact. Moreover, a replacement lemma, reduced to a one-block estimate, has to be established (see section \ref{sec:one}).

For any $G \in C_0 (\RR)$, we define  $T_{\gamma} G : \ZZ \to \RR$ by
$$(T_{\gamma} G)(x) =\sum_{j<x} (\gamma_j +\gamma_{j+1})\left\{ G \left( \cfrac{j+1}{N}\right) -G \left( \cfrac{j}{N}\right) \right\}$$
Observe that
$$N\cfrac{1}{\gamma_x + \gamma_{x+1}}\left[ (T_\gamma G)(x+1) -(T_\gamma G)(x)\right]= (\nabla_N G)(x/N)$$
where $\nabla_N$ stands for the discrete derivative: $(\nabla_N G)(x/N) =N\{G((x+1)/N) -G(x/N) \}$.

Since $T_{\gamma} G$ may not belong to $\ell_1 (\ZZ)$, we modify $T_\gamma G$ in order to integrate it with respect to the empirical measure. Fix $0<\theta<1/2$ and consider a $C^2$ increasing nonnegative  function ${\tilde g}$ defined on ${\mathbb R}$ such that ${\tilde g}(q)=0$ for $q \le 0$, ${\tilde g} (q)=1$ for $q \ge 1$ and ${\tilde g}(q) =q$ for $q\in [\theta, 1-\theta]$.

Fix an arbitrary integer $\ell >0$ and let $g=g_{\theta,\ell} : \RR \to \RR$ be given by
$$g(q)= {\tilde g} (q/\ell)$$

We define
\begin{equation*}
(T_{\gamma, \ell} G) (x) = (T_{\gamma} G)(x) -\cfrac{T_{\gamma,G}}{T_{\gamma,g}} (T_\gamma g)(x) 
\end{equation*}
where
\begin{equation*}
T_{\gamma,h} = \sum_{x \in \ZZ} (\gamma_x + \gamma_{x+1}) \left\{ h((x+1)/N) -h(x/N)\right\}
\end{equation*}

In the rest of the paper we make the choice $\ell:=\ell (N)=N^{1/4}$. 

\begin{Lemma}
\label{lem:JL}
For each function $G \in C_{0}^2 (\RR)$, and each environment $\gamma$ satisfying (\ref{eq:gounds}) and (\ref{eq:limg}), 
$$\lim_{N \to \infty} N^{1/4} \sup_{x \in \ZZ} \left| T_{\gamma,\ell} G (x) - {\bar \gamma} G(x/N)\right|=0 $$
and
$$\lim_{N \to \infty} N^{1/4} T_{\gamma,G} =0$$
\end{Lemma}

\begin{proof}
This is a slight modification of Lemma 4.1 in \cite{JL}. 
\end{proof}

We shall denote by $X_t^N \in {\mc M}$ the corrected empirical measure defined by 
\begin{equation*}
X_t^N (G) =X_{t}^{N,\gamma} (G) = \cfrac{1}{N} \sum_{x \in \ZZ} T_{\gamma,\ell} G (x) \, e_t (x)
\end{equation*}

The system is non-gradient but  we have 
\begin{equation*}
j_{x,x+1}= -\cfrac{1}{\gamma_x +\gamma_{x+1}} \nabla \left[ e_x +\cfrac{1}{2}\eta_{x-1}\eta_{x+1}\right] + {\mathcal L} \left( \cfrac{1}{2(\gamma_x +\gamma_{x+1})}\eta_x \eta_{x+1}\right)
\end{equation*}

This  implies that
\begin{eqnarray*}
N^2 {\mathcal L} \left[ X^N (G)\right]&=&\cfrac{1}{N} \sum_{x \in \ZZ} \left[ (\Delta_N G) (x/N)-\cfrac{T_{\gamma,G}}{T_{\gamma,g}} (\Delta_N g)(x/N)\right] \left(e_x+\cfrac{1}{2}\eta_{x-1}\eta_{x+1} \right)\\
&+&\cfrac{1}{2} \, {\mathcal L}\, \left( \sum_{x \in \ZZ} \left[ (\nabla_N G )(x/N) -\cfrac{T_{\gamma,G}}{T_{\gamma,g}}( \nabla_N g)(x/N)\right]\eta_x \eta_{x+1}\right)
\end{eqnarray*}
where $\Delta_N$ stands for the discrete Laplacian: 
$$(\Delta_N G)(x/N)= N^{2}\left\{ G((x+1)/N)+G((x-1)/N) - 2 G(x/N) \right\}$$
Therefore, we have
\begin{equation}
\label{eq:decompo}
X_t^N (G) -X_0^N (G)= U_t^N (G)+V_t^N (G)+M_t^N (G)
\end{equation}
with $M^N (G)$ a martingale and $U^N (G)$, $V^N (G)$, which are given by
\begin{equation*}
 U_t^N (G) = \int_0^t ds \; \cfrac{1}{N} \sum_{x \in \ZZ} B_N^G (x/N)\left(e_x (s) +\cfrac{1}{2}\eta_{x-1} (s)\eta_{x+1} (s)  \right)
\end{equation*}
where
$$B_N^G (x/N) =\left[ (\Delta_N G) (x/N)-\cfrac{T_{\gamma,G}}{T_{\gamma,g}} (\Delta_N g)(x/N)\right] $$
and
\begin{equation*}
V_t^N (G)= \cfrac{1}{N^2}\sum_x \left[ (\nabla_N G )(x/N) -\cfrac{T_{\gamma,G}}{T_{\gamma,g}}( \nabla_N g)(x/N)\right](\eta_x (t) \eta_{x+1}(t)- \eta_x(0)\eta_{x+1}(0))
\end{equation*}

\begin{Lemma}
The sequence $\left\{ \left( X_{\cdot}^N , \int_0^{\cdot} \pi_s^N ds \right) \in D ([0,T], {\mc M} ) \times D([0,T], {\mc M}^+ ) \, ; \, N \ge 1 \right\}$ is tight. 
\end{Lemma}

\begin{proof}
It is well known that the sequence $$\left\{ \left( X_{\cdot}^N , \int_0^{\cdot} \pi_s^N ds \right) \in D([0,T], {\mc M} ) \times D([0,T], {\mc M}^+ )\, ;\, N \ge 1\right\}$$ is tight if and only if the sequence $$\left\{ \left( X_{\cdot}^N (G) , \int_0^{\cdot} \pi_s^N (H) ds \right) \in D([0,T], \RR) \times D([0,T], \RR) \, ;\, N \ge 1\right\}$$ is tight for every $G,H \in C_0^2 (\RR)$. 

By Aldous criterion for tightness in $D([0,T], {\mathbb R})^2$, it is sufficient to show that

\begin{enumerate}
\item For every $t \in [0,T]$ and every $\varepsilon >0$, there exists a finite constant $A>0$ such that
$$\sup_N {\mathbb P}_{\mu^N} \left( \left| Y_t^N (G) \right| \ge A \right) \le \varepsilon$$

\item For every $\delta>0$,
$$\lim_{\varepsilon \to 0} \limsup_{N \to \infty} \sup_{ \tau \in \Theta , \theta \le \varepsilon} {\mathbb P}_{\mu^N} \left[ \left| Y^N_{\tau + \theta}(G) -Y_{\tau}^N (G)\right|\ge \delta\right]=0$$
where $\Theta$ is the set of all stopping times bounded by $T$.
\end{enumerate}
for  $Y_{\cdot}^N (G)=X_{\cdot}^N (G)$ and $Y_{\cdot}^N (G)=\int_0^{\cdot} \pi_s^N (G)ds$.

Since $G$ has compact support, there exists a constant $K>0$ (independent of $t$ and $N$) such that
\begin{eqnarray}
\label{eq:compa0}
{\mathbb E}_{\mu^N} \left[\left| {\bar \gamma} \langle \pi_t^N, G \rangle -X_t^N (G) \right| \right]\\
\le \left( N^{1/4} \sup_{x \in \ZZ} |{\bar \gamma} G(x./N) -(T_{\gamma,\ell} G) (x) | \right) \, {\mathbb E}_{\mu^N} \left[ \cfrac{1}{N^{5/4}} \sum_{|x| \le K  N^{5/4}} e_x (t) \right] \nonumber
\end{eqnarray}
and consequently
\begin{eqnarray}
\label{eq:compa}
{\mathbb E}_{\mu^N} \left[\int_0^T dt  \left| {\bar \gamma} \langle \pi_t^N, G \rangle -X_t^N (G) \right| \right]\\
\le \left( N^{1/4} \sup_{x \in \ZZ} |{\bar \gamma} G(x./N) -(T_{\gamma,\ell} G) (x) | \right)\int_0^T dt \, {\mathbb E}_{\mu^N} \left[ \cfrac{1}{N^{5/4}} \sum_{|x| \le K  N^{5/4}} e_x (t) \right] \nonumber
\end{eqnarray}
By Lemma \ref{lem:JL} and Lemma \ref{lem:mom-bounds}, the right-hand side of (\ref{eq:compa0}) (resp. of (\ref{eq:compa})) vanishes as $N \to \infty$. Hence, it is sufficient to show Aldous criterion for $Y_{\cdot}^N (G)=X_{\cdot}^N (G)$ and for $Y_{\cdot}^N (G)=\int_0^{\cdot} X_s^N (G)ds$. 


From the definition of the Skorohod topology, it is easy to show that the application $\Phi$ from $D([0,T],\RR)$ onto itself defined by
$$\Phi: x:= \left\{ x(t)\, ; \, 0 \le t \le T \right\} \to \Phi (x) := \left\{ \int_0^t x(s) ds \, ; \, 0 \le t \le T \right\}$$
is continuous. Thus, if $(X_{\cdot}^N (G))_N$ is tight, then $(\int_0^{\cdot} X_{s}^N (G))_{N}$ is tight.

Therefore it just remains to show Aldous criterion for $Y_{\cdot}^N (G)=X_{\cdot}^N (G)$.

{\textit{Proof of (1) for $X^N_{\cdot} (G)$:}}

\begin{equation*}
{\mathbb P}_{\mu^N} \left[ |X_t^{N} (G)| \ge A \right] \le \cfrac{1}{A}\, {\mathbb E}_{\mu^N} \left(\cfrac{1}{N} \sum_{x \in \ZZ} \left| (T_{\gamma,\ell} G)(x)\right | e_x (t) \right)  
\end{equation*}
We write $(T_{\gamma,\ell} G ) (x) = ((T_{\gamma,\ell} G ) (x)-{\bar \gamma} G (x/N)) + {\bar \gamma} G(x/N)$ and we get that
\begin{eqnarray*}
{\mathbb P}_{\mu^N} \left[ |X_t^{N} (G)| \ge A \right] &\le& \cfrac{1}{A}\, {\mathbb E}_{\mu^N} \left(\cfrac{1}{N } \sum_{|x| \le K N^{5/4}} \left| (T_{\gamma,\ell} G)(x) -{\bar \gamma} G(x/N) \right | e_x (t) \right) \\
&+& \cfrac{{\bar \gamma}}{A}\, {\mathbb E}_{\mu^N} \left(\cfrac{1}{N} \sum_{|x| \le K N} \left| G(x/N) \right | e_x (t) \right) 
\end{eqnarray*}
The first term on the right-hand side of the previous inequality can be bounded above by the right-hand side of (\ref{eq:compa0}), which vanishes. By Lemma \ref{lem:mom-bounds}, the second term is bounded above by $C/A$ with a constant $C$ independent of $N$. Therefore, the first condition is satisfied.

{\textit {Proof of (2) for $X^N_{\cdot} (G)$}:}

Recall the decomposition (\ref{eq:decompo}). In order to estimate the term 
$${\mathbb E}_{\mu^N} \left[ \left| U_{\tau+\ve}^N (G) -U_{\tau} (G) \right| \right]$$
we observe that $|B_N^G (x/N)|$ is bounded above by 
\begin{eqnarray*}
C \left[ {\bf 1}_{|x| \le KN} +\cfrac{\left|T_{\gamma,G}\right|}{\left|T_{\gamma,g}\right|} \left( \ell^{-2} + (N\ell^{-3})\right){\bf 1}_{\{x/(N\ell) \in [1-2\theta, 1+2 \theta] \cup [-2\theta, 2\theta]\}}\right]
\end{eqnarray*}
where $C,K$ are constants depending on $\theta$ and $G$ but not on $N$. By Schwarz inequality,  we are reduced to estimate
\begin{equation*}
 {\mathbb E}_{\mu^N} \left[ \int_{0}^{T+ \varepsilon} ds \; \cfrac{1}{N} \sum_{|x| \le 2 K N} e_x (s) \right]
\end{equation*}
and
\begin{equation*}
 \cfrac{\left|T_{\gamma,G}\right|}{\left|T_{\gamma,g}\right|} \left( \cfrac{1}{\ell^2} + \cfrac{1}{N\ell^3} \right){\mathbb E}_{\mu^N} \left[ \int_{0}^{T + \varepsilon} ds \; \cfrac{1}{N} \sum_{|x| \le (1+3\theta) N\ell} e_x (s) \right]
\end{equation*}
By Lemma \ref{lem:mom-bounds}, the first term is of order one. It is not difficult to show that  $\liminf_{N \to \infty} T_{\gamma,g} >0$, and Lemma \ref{lem:JL} gives $\ell T_{\gamma,G} \to 0$. Thus, by Lemma \ref{lem:mom-bounds}, the second one vanishes as $N$ goes to infinity.

The two last terms of (\ref{eq:decompo}) are given by
\begin{eqnarray*}
&\int_0^t ds \,{\mathcal L} \,\left( \sum_{x} \left[ (\nabla_N G )(x/N) -\cfrac{T_{\gamma,G}}{T_{\gamma,g}}( \nabla_N g)(x/N)\right]\eta_x \eta_{x+1}\right)(s)\\
&=V_t^N (G) +M_t^N (G)
\end{eqnarray*}
By using Lemma \ref{lem:JL} and Lemma \ref{lem:mom-bounds}, similar estimates as before show that
\begin{equation*}
\lim_{N \to \infty} \sup_{t \in [0,T+\varepsilon]} {\mathbb E}_{\mu^N} \left[ \left|V_{t}^N (G) \right| \right] =0
\end{equation*}

By computing the quadratic variation of the martingale $M^N (G)$, one obtains that (we recall that $\sup_x \gamma_x \le \gamma_+ $)
\begin{eqnarray*}
&{\mathbb E}_{\mu^N} \left[ \left( M_{\tau +\varepsilon}^N (G) -M^N_{\tau}(G)\right)^2\right]&\\
\le& \cfrac{2\gamma_+ }{N^2} \; {\mathbb E}_{\mu^N} \left[ \int_{\tau}^{\tau +\varepsilon} ds \sum_x  \left\{(\nabla_N G )(x/N) -\cfrac{T_{\gamma,G}}{T_{\gamma,g}}( \nabla_N g)(x/N) \right\}^2 e_s (x) e_{s} (x+1)\right]&
\end{eqnarray*}

Observe that 
$$\left\{(\nabla_N G )(x/N) -\cfrac{T_{\gamma,G}}{T_{\gamma,g}}( \nabla_N g)(x/N) \right\}^2 \le C {\bf 1}_{|x| \le K N}+\ell^{-2} (T_{\gamma,G}/T_{\gamma,g})^2 {\bf 1}_{|x| \le K N^{5/4}}$$

Since $\liminf_{N \to \infty} T_{\gamma,g} >0$ and $\ell T_{\gamma,G} \to 0$, from Lemma \ref{lem:mom-bounds}, we get
\begin{equation*}
\sup_{s \ge 0} \; {\mathbb E}_{\mu^N}\left\{\cfrac{1}{N^2} \sum_{|x| \le KN} e_x^2 (s)\right\} 
\end{equation*}
and
\begin{equation*}
\sup_{s \ge 0} {\mathbb E}_{\mu^N}\left\{\cfrac{1}{N^3} \sum_{|x| \le KN^{5/4}} e_x^2 (s)\right\} 
\end{equation*}
go to $0$ with $N$ (and are in particular bounded above by a constant independent of $N$).
\end{proof}

\begin{Lemma}
Let $(\alpha,\beta) \in {\mc M}  \times {\mc M}^+ $ be a limit point of the sequence 
$$\left\{ \left( X_{\cdot}^N , \int_0^{\cdot} \pi_s^N ds \right) \in D([0,T], {\mc M} ) \times D([0,T], {\mc M}^+ ) \, ;\, N \ge 1\right\}.$$ 

For every $G \in C_0^{2} (\RR)$ and every $t \in [0,T]$, we have
\begin{equation*}
\alpha_{t} (G) - \alpha_0 (G) = {\bar \gamma}^{-1} \int_{0}^{t} \alpha_s (\Delta G) ds, \quad \beta_t = {\bar \gamma}^{-1} \int_{0}^{t} \alpha_s ds
\end{equation*}
\end{Lemma}

\begin{proof}

In the proof of the tightness of $X_{\cdot}^N$ we have seen that the term
$$ {\mathbb E}_{\mu^N} \left[\int_0^t ds \; {\mathcal L} \left( \sum_{x} \left[ (\nabla_N G )(x/N) -\cfrac{T_{\gamma,G}}{T_{\gamma,g}}( \nabla_N g)(x/N)\right]\eta_x \eta_{x+1}\right)(s) \right]$$
and the term
\begin{equation*}
 {\mathbb E}_{\mu^N} \left[ \int_{0}^{t} ds \; \cfrac{1}{N} \sum_{x \in \ZZ} \cfrac{T_{\gamma,G}}{T_{\gamma,g}}( \Delta_N g)(x/N)\left(e_x (s) +\cfrac{1}{2}\eta_{x-1} (s)\eta_{x+1} (s)  \right)\right]
\end{equation*}
vanish as $N \to \infty$. By using Lemma \ref{lem:1b}, it implies that
\begin{equation*}
\alpha_{t} (G) - \alpha_0 (G) = \beta_t (\Delta G) 
\end{equation*}
Moreover, by (\ref{eq:compa}), we have
\begin{equation*}
\beta_t ={\bar \gamma}^{-1} \int_{0}^{t} \alpha_s ds
\end{equation*}
\end{proof}

\begin{Lemma}
Any limit point $\beta$ of the sequence $\{ \int_0^{\cdot} \pi_s^N ds \in D([0,T], {\mc M^+} )\, ; \, N \ge 1\}$ is such that, for any $t \in [0,T]$, $\beta_t$ is absolutely continuous with respect to the Lebesgue measure on $\RR$. 
\end{Lemma}

\begin{proof}
Fix a positive time $t$ and let $R_{\mu^N}$ be the probability measure on ${\mc M}^+$ given by
\begin{equation*}
R_{\mu^N} (A) = {\mathbb P}_{\mu^N} \left\{ \cfrac{1}{t} \int_0^t \pi_s^N ds \in A  \right\}
\end{equation*} 
for every Borel subset $A$ of ${\mc M}^+$. Let $J:{\mc M}^+ \to [0, +\infty)$ be a continuous and bounded function. By the entropy inequality (\ref{eq:enti}) and by using (\ref{eq:enttime}) we have
\begin{equation}
\label{eq:LVV}
\int J(\pi) dR_{\mu^N} (\pi)  \le C_0 +\cfrac{1}{N} \log \left( \int e^{N J(\pi) }dR_{\mu_{\bar \beta}} (\pi) \right)
\end{equation} 
By the Laplace-Varadhan theorem, the second term on the right hand side converges as $N$ goes to infinity to
\begin{equation*}
\sup_{\pi \in {\mc M}^+} \left[ J(\pi) -I_0 (\pi)\right]
\end{equation*}
where $I_0$ is the large deviations rate function for the random measure $\pi$ under $R_{\mu_{\bar \beta}}$. It is a simple exercise to compute the rate function $I_0$. We have
\begin{equation*}
I_{0} (\pi) = \sup_{f \in C_0 (\RR)} \left\{ \int f(u) \pi (du) - \int \log M_{\bar \beta} (f(u)) du  \right\}
\end{equation*}
where $M_{\bar \beta} (\alpha)$ is the Laplace transform of $\eta_0^2 /2$ under ${\mu}_{\bar \beta}$:
$$ M_{\bar \beta} (\alpha)=\mu_{\bar \beta} (e^{\alpha \eta_0^2 /2}) = \sqrt{{\bar \beta}/({\bar \beta - \alpha})}$$
if $\alpha < \bar \beta$, and $+ \infty$ otherwise.

The function $I_0$ also takes the simple form
\begin{equation*}
I_0 (\pi) =
\begin{cases}
\int_{\RR} h(\pi(u)) du \; {\rm{if}} \; \pi (du) =\pi(u) du,\\
+ \infty \; \rm{otherwise}
\end{cases}
\end{equation*}
where the Legendre transform $h$ of $M_{\bar \beta}$ is given by $h(\alpha)= {\bar \beta} \alpha -{1}/2 -{1/2} \log (2 \alpha {\bar \beta}) \ge 0$ if $\alpha >0$, and $+ \infty$ otherwise. 

Let $(f_k)_{k \ge 1}$ be a dense sequence in $C_0 (\RR)$ with $f_1$ being the function identically equal to $0$. Then $I_0$ is the increasing limit of $J_k \ge 0$ defined by
 \begin{equation*}
J_{k} (\pi) = \sup_{1\le j \le k} \left\{ \int f_j (u) \pi (du) - \int \log M_{\bar \beta} (f_j (u)) du  \right\} \wedge k
\end{equation*}

By using (\ref{eq:LVV}) we have
\begin{equation*}
\limsup_{N \to \infty} \int J_k (\pi) dR_{\mu^N} (\pi) \le C_0
\end{equation*}
for each $k$. Since $J_k$ is a lower semi-continuous function, any limit point $R^*$ of $R_{\mu^N}$ is such that
\begin{equation*}
\int J_k (\pi) dR^* (\pi) \le C_0
\end{equation*}
By the monotone convergence theorem, we have $\int I_0 (\pi) dR^{*} (\pi) \le C_0 < +\infty$. Since $I_0 (\pi)$ is equal to $+ \infty$ if $\pi$ is not absolutely continuous with respect to the Lebesgue measure, it implies that 
$$R^* \left\{ \pi; \pi (du) =\pi (u) du \right\}=1$$
and the lemma is proved.
\end{proof}

We conclude as follows. Let $(\alpha,\beta)$ be a limit point of $(X^N_{\cdot} (G), \int_0^{\cdot} \pi_s^N (G))_{N\ge 1}$. From the equation 
$$\alpha_\cdot (G)  - \alpha_0 (G) ={\bar \gamma}^{-1} \int_0^{\cdot} \alpha_s (\Delta G) ds$$ 
we see that $\alpha$ is time continuous. Moreover, if $A$ is a subset of $\RR$ with zero Lebesgue measure, then $\beta_t (A) =0$ for any $t \in [0,T]$. This implies that $\alpha_t (A)=0$ for any $t \in [0,T]$, i.e. that $\alpha_t$ is absolutely continuous with respect to the Lebesgue measure on $\RR$.

By uniqueness of weak solution to the heat equation, we have that $2 \alpha$ is the Dirac mass concentrated on the (smooth) solution of the heat equation $(t,u) \in [0,T] \times \RR \to {\tilde T}_{t} (u) $ starting from ${\bar \gamma}\beta_0^{-1}$: ${\partial_t} {\tilde T}_t = {\bar \gamma}^{-1} \Delta {\tilde T}_t, \; {\tilde T}_0 = {\bar \gamma} {\beta_0}^{-1}$.

Hence we conclude that $\{ X^N_{\cdot} \in D([0,T], {\mc M}) \, ; \, N \ge 1\}$ converges in distribution to $({\tilde T}_\cdot (u) /2) \, du$. Since the limit is continuous in time we have that $\{X_t^N\, ; \, N \ge 1 \}$ converges in distribution to the deterministic limit $({\tilde T}_t (u) /2) \,  du$. Since convergence in distribution to a deterministic variable implies convergence in probability, this implies that 
$$\lim_{N \to \infty} {\mathbb P}_{\mu^N}\left[ \left |X_{t}^N (G) -\cfrac{1}{2}\int {\tilde T}_t (u) G(u) du \right| \ge \varepsilon \right] =0$$
We use again (\ref{eq:compa0})  and the fact that ${\bar \gamma}^{-1} {\tilde T}_t =T_t$ to get
$$\lim_{N \to \infty} {\mathbb P}_{\mu^N} \left[ \left |\pi_t^N (G) -\cfrac{1}{2} \int T_t (u) G(u) du \right| \ge \varepsilon \right] =0$$
and the theorem is proved.

\section{One-block estimate}
\label{sec:one}

The aim of this section is to prove the following so-called one block estimate (\cite{KL}).

\begin{Lemma}[One block estimate] 
\label{lem:1b}
For any $G \in C_0^2 (\RR)$, any $t \ge 0$, and any $\delta>0$,  
\begin{equation*}
\lim_{N \to \infty} {\mathbb P}_{\mu^N} \left[ \left| \cfrac{1}{N} \sum_{x \in \ZZ} (\Delta_N G) (x/N) \int_0^t ds\,  \eta_{x-1} (s) \eta_{x+1}(s) \right| \ge \delta \right] =0
\end{equation*}
\end{Lemma}

Since $G \in C^2_0 (\RR)$, we can replace $(\Delta_N G)(x/N)$ by
$$\cfrac{1}{2k+1} \sum_{|y-x| \le k} (\Delta_N G)(y/N)$$
as soon as $k \ll N$ and we are left to prove that
\begin{equation*}
\lim_{k \to \infty}\lim_{N \to \infty} {\mathbb E}_{\mu^N} \left[ \cfrac{1}{2N+1} \sum_{|x| \le N} \int_0^t ds\, \left|\cfrac{1}{2k+1} \sum_{|x-y| \le k} \eta_{y}(s)\eta_{y+1} (s) \right|  \right]=0
\end{equation*}

Given two probability measures $P,Q$ on $\Omega$ and $\Lambda$ a finite subset of $\ZZ$, $H_{\Lambda} (P | Q)$ denotes the relative entropy of the projection of $P$ on ${\mathbb R}^{\Lambda}$ with respect to the projection of $Q$ on ${\mathbb R}^{\Lambda}$. We shall denote the projection of $P$ on $\RR^{\Lambda}$ by $P |_{\Lambda}$. If $\Lambda=\Lambda_k= \{ -k,\ldots,k\}$, we use the short notation $P_k$.  

We define the space-time average of $(\mu_s^N)_{0 \le s \le t}$ by
\begin{equation*}
  \nu^{N} = \frac{1}{(2N+1)t} \sum_{|x| \le N} \int_0^t \tau_x \mu_s^N \; ds
\end{equation*}

Here $\tau_x$ denotes the shift by $x$: for any $\eta \in \Omega$, the configuration $\tau_x \eta$ is defined by $(\tau_x \eta)_{z} =\eta_{x+z}$; for any function $g$ on $\Omega$, $\tau_x g$ is the function on $\Omega$ given by $(\tau_x g) (\eta)= g(\tau_x \eta)$; for any $p \in {\mc P} (\Omega)$, $\tau_x p$ is the push-forward of $p$ by $\tau_x$. The probability measure $p$ is said to be translation invariant if $\tau_x p =p$ for any $x \in \ZZ$.

We have to show
\begin{equation}
\label{eq:finfin}
\lim_{k \to \infty}\lim_{N \to \infty} \int_{\RR^{\Lambda_k} }d{\nu}_k^N \left[  \left|\cfrac{1}{2k+1} \sum_{|y| \le k} \eta_y \eta_{y+1} \right|  \right]=0
\end{equation}

\begin{Lemma}
  For each fixed $k$, the sequence of probability measure $(\nu^N_k)_{N \ge k}$ on $\RR^{\Lambda_k}$ is tight. 
\end{Lemma}

\begin{proof}
It is enough to prove that there exists a constant $C_k<\infty$ independent of $N$ such that
\begin{equation}
\label{eq:bounden}
  \int \sum_{i \in \Lambda_k} e_i \, d\nu_k^N \le C_k
\end{equation}

We begin to prove that
\begin{equation}
\label{eq:entropo}
H_{\Lambda_k} (\nu^N | \mu_{\bar \beta}) = H \left(\nu_k^N\, \Big|\,  \mu_{\bar \beta} |_{\Lambda_k} \right) \le C_0 |\Lambda_k|
\end{equation} 

Fix a bounded measurable function $\phi$ depending only on the sites in $\Lambda:=\Lambda_k=\{-k, \ldots,k\}$. Assume for simplicity that $2N+1=(2k+1)(2p+1)$ for some $p \ge 1$. Then we can index the elements of the set $\{-N,\ldots,N\}$ in the following way
$$\{-N, \ldots,N\}= \{ x_j +y; j=-p, \ldots, p\,;\, y \in \Lambda_k\}$$
where $x_j = 2k j +1$. Since $\phi$ depends only on the sites in $\Lambda$, it is clear that under $\mu_{\bar \beta}$, for each $y \in \Lambda$, the random variables $\left( \tau_{x_j +y} \phi \right)_{j=-p,\ldots,p}$ are i.i.d.. 

Let ${\bar \mu}_t^N=t^{-1} \int_0^t \mu_s^N ds$. By convexity of the entropy and (\ref{eq:enttime}), we have $H({\bar \mu}_t^N | \mu_{\bar \beta}) \le C_0 N$.

We write
\begin{eqnarray*}
\int \phi d\nu_k^N &=& \int \left( \cfrac{1}{2N+1} \sum_{|x| \le N} \tau_x \phi \right) d{\bar \mu}_t^N\\
&\le&\cfrac{|\Lambda|}{2N+1} H({\bar \mu}_t^N | \mu_{\bar \beta})+ \cfrac{|\Lambda|}{2N+1}\log \left( \int d\mu_{\bar \beta} \, e^{|\Lambda|^{-1} \sum_{|x| \le N} \tau_x \phi}\right)\\
&\le& C_0 |\Lambda| +  \cfrac{|\Lambda|}{2N+1}\log \left( \int d\mu_{\bar \beta} \, e^{|\Lambda|^{-1} \sum_{ y \in \Lambda} (\sum_{|j| \le p} \tau_{y+x_j} \phi)}\right)\\
&\le& C_0 |\Lambda| + \cfrac{|\Lambda|}{2N+1}|\Lambda|^{-1} \sum_{ y \in \Lambda}\log \left( \int d\mu_{\bar \beta} \, e^{ \sum_{|j| \le p} \tau_{y+x_j} \phi}\right)
\end{eqnarray*} 
where we used the entropy inequality (\ref{eq:enti}) and the convexity of the application $ f \to \log \left( \int d\mu_{\bar \beta} \, e^f \right)$. By independence, for each $y$, of $\left( \tau_{x_j +y} \phi \right)_{j=-p,\ldots,p}$ and the translation invariance of $\mu_{\bar \beta}$, we get
\begin{equation}
\int \phi d\nu_k^N \le C_0 |\Lambda| + \log \left(\int e^\phi d\mu_{\bar \beta}\right)
\end{equation} 
This implies (\ref{eq:entropo}) and, by the entropy inequality (\ref{eq:enti}), the inequality (\ref{eq:bounden}).
\end{proof}

For any $k$, let $\nu^*_k$ be a limit point of the sequence $(\nu_k^N)_{N\ge k}$. The sequence of probability measures $(\nu^*_k)_{k \ge 0}$ forms a consistent family and, by Kolmogorov theorem, there exists a unique probability measure $\nu$ on $\Omega$ such that $\nu_k=\nu^*_k$. By construction, the probability measure $\nu$ is invariant by translations.

\begin{Lemma}
There exists  $C_0$ such that for any box $\Lambda_k=\{-k, \ldots,k\}$, $k \ge 0$, 
\begin{equation}
\label{eq:dyn-678}
H_{\Lambda_k} (\nu | \mu_{{\bar \beta}}) \le C_0 |\Lambda_k| 
\end{equation}
\end{Lemma}

\begin{proof}
We have seen in the proof of the previous lemma that
\begin{equation}
H_{\Lambda_k}(\nu^N | \mu_{\bar \beta})=H \left(\nu_k^N \; \Big| \;  \mu_{\bar \beta}|_{\Lambda_k} \right) \le C_0 |\Lambda_k|
\end{equation} 
Since the entropy is lower semicontinuous, it follows that
$$H_{\Lambda_k}(\nu | \mu_{\bar \beta}) \le C_0 |\Lambda_k|$$
\end{proof}

A translation invariant probability measure $\nu$ on $\Omega$ such that (\ref{eq:dyn-678}) is satisfied is said to have a finite entropy density. By a super-additivity argument (see \cite{BOO}, \cite{FFL}), the following limit
\begin{equation}
{\bar H} (\nu| \mu_{{\bar \beta}}) =\lim_{k \to \infty} \cfrac{H_{\Lambda_k} (\nu | \mu_{{\bar \beta}}) }{|\Lambda_k|}
\end{equation} 
exists and is finite. 
For any bounded local measurable function $\phi$ on $\Omega$, we define the limit
\begin{equation*}
{\bar F} (\phi) =\lim_{k \to \infty} \cfrac{1}{2k+1} {\bar F}_k (\phi), \quad {\bar F}_k (\phi) = \log \int e^{\sum_{i=-k}^{k} \tau_i \phi} d\mu_{\bar \beta}
\end{equation*}

The entropy density ${\bar H}(\nu |\mu_{\bar \beta})$ can be expressed by the variational formula
\begin{equation}
\label{eq:vfh}
{\bar H} (\nu | \mu_{{\bar \beta}}) = \sup_{\phi} \left\{ \int \phi d\nu -{\bar F} (\phi) \right\} 
\end{equation}
where the supremum is taken over all bounded local measurable functions $\phi$ on $\Omega$.



We now show the following lemma
\begin{Lemma}
For any function $F \in C_0^1 (\Omega)$, we have
\begin{equation*}
\int {\mc L}F \, d\nu =0
\end{equation*} 
\end{Lemma}

\begin{proof}
Assume that $F\in C_0^1 (\Omega)$ has a support included in ${\mathbb R}^{\Lambda_{k-1}}$. 
We have
\begin{equation*}
\int {\mc L}F \, d\nu = \int {\mc L}F \, d\nu_k = \lim_{N \to \infty} \int {\mc L}F \, d\nu_k^N
\end{equation*} 

Define $G=(2N+1)^{-1} \sum_{|x| \le N} \tau_x F$. By It\^o formula
\begin{equation*}
\label{eq:Itoo1}
\begin{split}
&N^{-2} \left\{ \int d\mu_t^N (\eta) G (\eta) - \int d\mu^N (\eta) G(\eta)\right\} \\
&= \int_0^t ds \int d\mu_s^N (\eta) \, ({\mc L}G) (\eta)\\
&= t \int  d{\bar \mu}_t^N (\eta) \, ({\mc L}G) (\eta) \nonumber \\
&= t \int d\nu_k^N (\eta) \, ({\mc L} F)(\eta) \nonumber
\end{split}
\end{equation*}
Since $F$ (and hence $G$) is bounded, the left-hand side goes to $0$ as $N$ goes to infinity and it follows that
\begin{equation*}
\int {\mc L} F \, d\nu =0
\end{equation*}
\end{proof}

Recall that we want to show (\ref{eq:finfin}). From the previous lemmas, it is sufficient to prove that
\begin{equation*}
\lim_{k \to \infty} \int d{p} (\eta) \left[  \left|\cfrac{1}{2k+1} \sum_{|y| \le k} \eta_y \eta_{y+1} \right|  \right]=0
\end{equation*}
for any $p \in {\mc P} (\Omega)$ such that $p$ has finite entropy density, is stationary for ${\mc L}$ and translation invariant. 

Proposition \ref{prop:ergo}  gives the characterization of stationary probability measures, translation invariant, and with finite entropy density. By using the notations of this proposition, to complete the proof of Lemma \ref{lem:1b},  we have to show that
\begin{equation*}
\lim_{k \to \infty} \int_{(0, + \infty)} d\lambda (\beta) \int_{\RR^{2k+1}} d{\mu_{\beta}} \left[  \left|\cfrac{1}{2k+1} \sum_{|y| \le k} \eta_{y}\eta_{y+1} \right|  \right]=0
\end{equation*}
Since, under $\mu_\beta$, the random variables $(\sqrt{\beta} \eta_y)_{y}$ are distributed according to standard independent Gaussian variables, and $\int \beta^{-1} d\lambda (\beta) < +\infty$, it remains to prove 
\begin{equation*}
\lim_{k \to \infty} \int_{\RR^{2k+1}} d{\mu_{1}} \left[  \left|\cfrac{1}{2k+1} \sum_{|y| \le k} \eta (y)\eta (y+1) \right|  \right]=0
\end{equation*}
By using Schwarz inequality, a simple computation gives the result.

\begin{prop}
\label{prop:ergo}
Let $\nu$ be an invariant measure for ${\mathcal L}$ which is translation invariant with finite entropy density. Then, $\nu$ is a mixture of the Gaussian product measures $\mu_{\beta}$, $\beta>0$,
$$\nu = \int_{(0,+\infty)} d\lambda (\beta) \, \mu_\beta$$
and the probability measure $\lambda$ on $(0,+\infty)$ is such that
$$\int_{(0,+\infty)} \beta^{-1} d\lambda (\beta) < + \infty$$
\end{prop}

In order to give the proof of this proposition, we need the following lemma

\begin{Lemma}
\label{lem:lem9}
Let $\nu$ be an invariant measure for ${\mathcal L}$, translation invariant with finite entropy density. Then, for any local measurable bounded function $\phi$ on $\Omega$, we have 
\begin{equation*}
\forall x \in \ZZ, \quad \int \left[ \phi (\eta^x) - \phi (\eta)\right] d\nu =0
\end{equation*}
\end{Lemma}

\begin{proof}
We only give a sketch of the proof since the arguments are almost the same as in \cite{FFL}, Proposition 6.1 (see also chapter 2 of \cite{BOO}). 

The proof is divided in two steps. Let us first consider a generic probability measure $\nu_*$, not necessarily translation invariant, such that $H(\nu_* | \mu_{\bar \beta}) < +\infty$ and let us denote by $g$ the density of $\nu_*$ with respect to $\mu_{\bar \beta}$. We introduce, for any $n$, the Dirichlet forms
\begin{equation}
\label{eq:df-v}
D_n (\nu_*) = \sup_{\psi} \left\{ - \int \cfrac{{\mc S}_n \psi}{\psi} d\nu_*\right\}
\end{equation}
where the supremum is carried over the set ${\mc F}$ composed of the positive functions $\psi : \Omega \to (0,+\infty)$ such that $0<M^{-1} \le \psi \le M$ for some positive constant $M$. 

It is easy to check that if $D_n (\nu_*) <+\infty$ then 
\begin{equation}
\label{eq:last}
D_n (\nu_*) = \cfrac{1}{2} \sum_{x=-n}^n \gamma_x \int \left( Y_x {\sqrt g} \right)^2 d {\mu}_{\bar \beta}
\end{equation}
where for any function $u: \Omega \to \RR$, $Y_x u$ is the function defined by  $(Y_{x} u) (\eta) = u(\eta^x) -u(\eta)$. Observe that $Y_x^2 = -2 Y_x$ so that $-{\mc S}_n =(1/2) \sum_{x=-n}^n \gamma_x  Y_x^2$.

In fact, even if $\nu_*$ is not absolutely continuous with respect to $\mu_{\bar \beta}$, the Dirichlet form $D_n (\nu_*)$ defined by (\ref{eq:df-v}) makes sense in $[0,+\infty]$.

Recall that $(P_t^n)_{t \ge 0}$ is the semigroup generated by the finite dimensional dynamics introduced in section \ref{sec:model}.  We have the following well known entropy production bound (see \cite{BOO} or \cite{KL}, Theorem 9.2)
\begin{equation*}
H(\nu_* P_t^n \, | \, \mu_{\bar \beta} ) + t D_n ({\bar \nu}_{*,t}^n) \le H (\nu_* \, | \, \mu_{\bar \beta})
\end{equation*}
where ${\bar \nu}_{*,t}^n = t^{-1} \int_{0}^t \nu_* P_{s}^n ds$.

Let us denote the density of ${\bar \nu}_{*,t}^n $ with respect to $\mu_{\bar \beta}$ by ${\bar g}_t^n$. Since $H (\nu_* | \mu_{\bar \beta}) <+\infty$, we have $D_n ({\bar \nu}_{*,t}^n) < +\infty$ and, by the explicit formula (\ref{eq:last}) of the Dirichlet form,
\begin{equation}
H(\nu_* P_t^n \, | \, \mu_{\bar \beta} ) + \cfrac{\gamma_- \,t}{2} \sum_{j=-n}^n \int \left( Y_j \sqrt{{\bar g}_t^n} \right)^2  d{\mu_{\bar \beta}} \le H (\nu_* \, | \, \mu_{\bar \beta})
\end{equation}
The second term on the left-hand side of the previous inequality is composed by a sum of positive parts. We can restrict this for any $m \le n$. By using (\ref{eq:vfh0}) and the variational formula (\ref{eq:df-v}) for the Dirichlet form, we get that, for any function $\phi \in C_0^1 (\Omega)$ and any functions  $\psi_j \in {\mc F}, j\in \{-m,\ldots,m\}$, 
\begin{equation*}
\int P_t^n \phi d\nu_* - \log \int e^{\phi} d \mu_{\bar \beta} + \cfrac{\gamma_- \,  t}{2} \sum_{j=-m}^m \cfrac{Y_j^2 \psi_j}{\psi_j}\, d {\bar \nu}_{*,t}^n  \le H (\nu_* | \mu_{\bar \beta})
\end{equation*}
We let $n \to \infty$ and, by (\ref{eq:10}), we have
\begin{equation}
\label{eq:fin}
\int P_t  \phi d\nu_* - \log \int e^{\phi} d \mu_{\bar \beta} + \cfrac{\gamma_- \, t}{2} \sum_{j=-m}^m \cfrac{Y_j^2 \psi_j}{\psi_j} \, d {\bar \nu}_{*,t}  \le H (\nu_* | \mu_{\bar \beta})
\end{equation}
where ${\bar \nu}_{*,t} = t^{-1} \int_0^t \nu_* P_s ds$.

In the second step of the proof we apply (\ref{eq:fin}) to $\nu_* = \nu_*^{(m)} = \nu \Big|_{\Lambda_m} \otimes\,  \mu_{\bar \beta} \Big |_{\Lambda_m^c}$. We recall that $\Lambda_m$ denotes the box $\{ -m, \ldots,m\}$ and $\Lambda_m^c$ stands for $\ZZ \backslash \Lambda_m$. Observe that $H(\nu_*^{(m)} | \mu_{\bar \beta}) = H_{\Lambda_m} (\nu | \mu_{\bar \beta})$ so that
\begin{equation*}
\lim_{m \to \infty} (2m+1)^{-1} H (\nu_*^{(m)} | \mu_{\bar \beta}) ={\bar H} (\nu | \mu_{\bar \beta})
\end{equation*}

By choosing $\phi = \sum_{i=-m}^{m} \tau_i \phi_0$, $\psi_i =\tau_i \psi_0$, with $\phi_0 \in C_0^1 (\Omega)$ and $\psi_0 \in {\mc F}$, we get
\begin{equation*}
\sum_{i=-m}^{m} \int P_t (\tau_i \phi_0) d\nu_{*}^{(m)} - {\bar F}_m (\phi_0) + \cfrac{ \gamma_- \, t}{2} \sum_{i=-m}^m \int \tau_i \cfrac{Y_0^2 \psi_0}{\psi_0 } d{\bar \nu}_{*,t}^{(m)} \le H (\nu_*^{(m)} | \mu_{\bar \beta}) 
\end{equation*}

We claim that
\begin{equation}
\label{eq:diff}
\begin{split}
\lim_{m \to \infty} \cfrac{1}{2m+1} \sum_{i=-m}^{m} \int P_t (\tau_i \phi_0) d\nu_{*}^{(m)} = \int P_t \phi_0 d\nu = \int \phi_0 d\nu , \\
\lim_{m \to \infty} \cfrac{1}{2m+1} \sum_{i=-m}^m \int \tau_i \cfrac{Y_0^2 \psi_0}{\psi_0 } d{\bar \nu}_{*,t}^{(m)} = \int \cfrac{Y_0^2 \psi_0}{\psi_0 } d{\nu} 
\end{split}
\end{equation}

Then, by using (\ref{eq:vfh}) and optimizing over $\phi_0$ and $\psi_0$, we get
\begin{equation*}
\sup_{\psi_0} \int \cfrac{Y_0^2 \psi_0}{\psi_0} d\nu =0
\end{equation*}
It is clear that we can repeat the argument substituting $Y_j$ to $Y_0$, and we obtain
\begin{equation*}
\sup_{\psi_0} \int \cfrac{Y_j^2 \psi_0}{\psi_0} d\nu =0
\end{equation*}
so that, by summing over $j$, we have $D_n (\nu)=0$ which implies that $\nu$ is invariant by any flip. 

It remains to show (\ref{eq:diff}). The difficulty comes from the fact that even if the function $u$ is local  $P_t u$ is not. But it is easy to see, by using (\ref{eq:10}), that we can replace the semigroup of the infinite dynamics $P_t$ by the semigroup of the finite dimensional dynamics $P_t^n$, if $n$ is sufficiently large. The function $P_t^n u$ is then local and the ergodic theorem permits to conclude. 

We refer the interested reader to \cite{BOO} for the details of the arguments.

\end{proof}

\begin{proof}[Proof of Proposition \ref{prop:ergo}]

By Lemma \ref{lem:lem9}, we have $\int {\mc S} g \, d \nu =0$ for any bounded measurable function $g$ on $\Omega$. It follows that for any $g \in C_0^1 (\Omega)$,
\begin{equation}
\label{eq:af}
\int {\mc A} g \, d\nu =0
\end{equation}

Since $\nu$ has finite entropy density, we have $\int e_{0} d\nu < +\infty$. By translation invariance, the ergodic theorem gives the existence $\nu$ a.s., and in ${\mathbb L}^{1} (\nu)$,  of 
\begin{equation*}
u(\eta)=\lim_{\ell \to \infty} \cfrac{1}{2\ell +1} \sum_{|x| \le \ell} \eta_x, \quad {\mc E}(\eta)= \lim_{\ell \to \infty} \cfrac{1}{2\ell +1} \sum_{|x| \le \ell} \eta_x^2
\end{equation*}

Since $\nu$ is invariant with respect to any flip, we have $\nu$ almost surely that $u(\eta)= 0$. 

Assume first that $\nu$ is exchangeable.

For any ${\bz} \in [0,\infty)$ let $\nu_{\bz}$ be the probability measure 
$$\nu_{\bz} =\nu\left(\cdot | {\mc E} =\bz \right)$$

If $\bz=0$ then $\nu_{\bz}$ is the Dirac mass concentrated on the configuration $\delta_0$ with each coordinate equal to $0$. 

Let us now assume that $\bz \ne 0$. 

Consider a test function $g$ in (\ref{eq:af}) of the form 
$$g(\eta) = f (\eta) \, \chi\left( \cfrac{1}{2\ell +1} \sum_{|x| \le \ell} \eta^2_x \right)$$
with $f, \chi$ compactly supported and smooth. It is easy to show, by taking the limit $\ell \to \infty$ in (\ref{eq:af}) with $g$ as above, that
\begin{equation*}
\int {\mc A} f \, d\nu_{\bz} =0
\end{equation*}
and this can be extended to any $f \in C^{1}_0 (\Omega)$. We apply the previous equality with a function $f$ of the form
\begin{equation*}
f(\eta)= \eta_x \phi (\eta)
\end{equation*}
with $\phi \in C_0^1 (\Omega)$ independent of $\eta_x$. Then we get
\begin{eqnarray*}
0&=&\int (\eta_{x+1} -\eta_{x-1} ) \phi d\nu_{\bz} (\eta) + \sum_{y \ne x} \int (\eta_{y+1} -\eta_{y-1}) {\eta_x} \partial_{\eta_y} \phi d\nu_{\bz}\\
&=& \int (\eta_{x+1} -\eta_{x-1} ) \phi d\nu_{\bz} (\eta)\\
&+&  \int (\eta_{x} -\eta_{x-2}) {\eta_x} \partial_{\eta_{x-1}} \phi d\nu_{\bz} +  \int (\eta_{x+2} -\eta_{x}) {\eta_x}  \partial_{\eta_{x+1}} \phi d\nu_{\bz}\\
&+&  \sum_{y \ne x-1,x,x+1} \int (\eta_{y+1} -\eta_{y-1}) {\eta_x} \partial_{\eta_y} \phi d\nu_{\bz}\\
\end{eqnarray*}

We claim that the last term is equal to zero. This is a consequence of the exchangeability of $\nu_{\bz}$. Let $\Lambda$ be the support of $\phi$ (which does not contain $x$ by assumption). Observe that, for any $y \ne x-1,x,x+1$, the site $x$ does not belong to the support of $(\eta_{y+1}- \eta_{y-1}) \partial_{\eta_y} \phi$. Let $t$ be sufficiently large (e.g. $t > |x| + \max_{s \in \Lambda} |s| +10$). By exchangeability we have, for any $k \ge 0$, that
\begin{equation*}
\int (\eta_{y+1} -\eta_{y-1}) {\eta_x} \partial_{\eta_y} \phi d\nu_{\bz}= \int (\eta_{y+1} -\eta_{y-1}) \eta_{t+k} \partial_{\eta_y} \phi d\nu_{\bz}
\end{equation*}    
Hence, we get
\begin{equation*}
\int (\eta_{y+1} -\eta_{y-1}) {\eta_x} \partial_{\eta_y} \phi d\nu_{\bz}= \cfrac{1}{\ell} \sum_{k=0}^{\ell-1} \int (\eta_{y+1} -\eta_{y-1}) \eta_{t+k} \partial_{\eta_y} \phi d\nu_{\bz}
\end{equation*}   
Let $\ell$ go to infinity and use the convergence of ${\ell}^{-1} \sum_{k=0}^{\ell-1} \eta_{t+k}$ to $u(\eta)=0$ to conclude.

The same argument shows that
\begin{equation*}
  \int \eta_{x-2} {\eta_x}  \partial_{\eta_{x-1}} \phi d\nu_{\bz} =0, \quad \int \eta_{x+2} {\eta_x}\partial_{\eta_{x+1}} \phi d\nu_{\bz} =0
\end{equation*}
and, similarly, we have
\begin{equation*}
 \int \eta_{x}^2 \partial_{\eta_{x-1}} \phi d\nu_{\bz} = \bz \int \partial_{\eta_{x-1}} \phi d\nu_{\bz}, \quad  \int \eta_{x}^2 \partial_{\eta_{x+1}} \phi d\nu_{\bz} = \bz  \int \partial_{\eta_{x+1}} \phi d\nu_{\bz},
\end{equation*}

Hence, we proved that, for any $x \in \ZZ$ and for any function $\phi \in C_0^1 (\Omega)$ such that $x$ does not belong to the support of $\phi$, 
\begin{equation*}
\int (\eta_{x+1} -\eta_{x-1} ) \phi d\nu_{\bz} (\eta) + \bz \int  (\partial_{\eta_{x-1}} -\partial_{\eta_{x+1}}) \phi d\nu_{\bz} =0
\end{equation*}

We apply this for a function $\phi$ depending only on $(\eta_{2k})_{k \in \ZZ}$, so that, for any $k$, $\phi$ is independent of $\eta_{2k+1}$. We have
\begin{equation*}
\int (\eta_{2k+2} -\eta_{2k} ) \phi d\nu_{\bz} (\eta) + \bz \int  (\partial_{\eta_{2k}} -\partial_{\eta_{2k+2}}) \phi d\nu_{\bz} =0
\end{equation*}
This implies that the law of $(\eta_{2k})_{k \in \ZZ}$ under $\nu_{\bz}$ is a product of centered Gaussian probability measures on $\RR$ with variance $\bz$ (see e.g. \cite{Hsu}).

The same result occurs for the law of $(\eta_{2k+1})_{k \in \ZZ}$.

Let now $\Phi_{A} (\eta) = \prod_{s \in A} \phi_s (\eta_s)$ be a test function with $A$ a finite arbitrary set of $\ZZ$ and $\phi_s$ real valued bounded functions. We write the set $A$ in the form $A_0 \cup A_1$ where $A_0$ is the set composed of elements of $A$ which are even and $A_1$ the set composed of elements of $A$ which are odd. Let $B_0$ be a set composed of even sites such that $|B_0|=|A_1|$ and $A \cap B_0 =\emptyset$. Let $\sigma$ be a permutation on $\ZZ$ such that $\sigma (A_1)=B_0$ and $A_0$ is fixed under the action of $\sigma$. We denote by $\sigma \cdot \eta$ the configuration defined by $(\sigma \cdot \eta)_x = \eta_{\sigma (x)}$. By exchangeability of $\nu_{\bz}$ we have
\begin{equation*}
\nu_{\bz} (\Phi_A (\sigma \cdot \eta)) =\nu_{\bz} (\Phi_A (\eta))
\end{equation*} 
and 
$$\Phi_A (\sigma \cdot \eta)= \prod_{s \in A_0} \phi_s (\eta_s) \prod_{s \in B_0} \phi_{\sigma^{-1} (s)} (\eta_s)$$
Since the function $\Phi_{A} (\sigma \cdot \eta)$ is a function depending only on $(\eta_{2k})_{k \in \ZZ}$ and $A_0 \cap B_0 = \emptyset$ we know that
\begin{equation*}
\nu_{\bz} (\Phi_{A} (\sigma \cdot \eta)) = \prod_{s \in A_0} \left(\int \phi_s (x) g_{1/\bz }(x) dx \right) \prod_{s \in B_0} \left( \int \phi_{\sigma^{-1} (s)} (x)  g_{1/\bz}(x) dx \right)
\end{equation*}
We recall that $g_{1/\bz}$ is the density of the centered Gaussian probability measure on $\RR$ with variance $\bz$.
Hence, we proved
\begin{equation*}
\nu_{\bz} (\Phi_A ( \eta)) = \prod_{s \in A} \left(\int \phi_s (x) g_{1/\bz}(x) dx \right) 
\end{equation*} 
which shows that $d \nu_{\bz} (\eta)$ is equal to $ \prod_{x \in \ZZ} g_{1/\bz} (\eta_x) d\eta_x$.

We now show that $\nu$ is exchangeable. Let us consider the test function $\chi(\eta)= \phi(\eta_x,\eta_{x+1}) \psi (e_y; y\ne x,x+1)$ with $\phi,\psi$ smooth and compactly supported functions. By (\ref{eq:af}), we have 
$$\int d\nu{\mc A} \chi =0= \int {\mc A}\phi\,  \psi \,d\nu  + \int {\phi} \, {\mc A}\psi\,  d\nu$$
Observe that the second term is given by
$$\sum_{y \ne x,x+1} \int d\nu (\eta) \eta_y (\partial_{e_y} \psi)(\eta) (\eta_{y+1}- \eta_{y-1}) \phi(\eta_x, \eta_{x+1})$$
This is equal to zero because $\nu$ is invariant by the flips and the function $\eta \to \eta_y (\partial_{e_y} \psi)(\eta) (\eta_{y+1}- \eta_{y-1}) \phi(\eta_x, \eta_{x+1})$ is an odd function of $\eta_y$ for $y \ne x,x+1$. 

Moreover we have that
$$({\mc A} \phi)(\eta)= (\eta_{x+2} - \eta_x) \partial_{\eta_{x+1}} \phi + (\eta_{x+1} -\eta_{x-1}) \partial_{\eta_x} \phi$$
Remark that $\eta_{x+2} \psi \partial_{\eta_{x+1}} \phi$ is odd with respect to $\eta_{x+2}$ so that its integral with respect to $\nu$ is equal to $0$, and similarly for $\eta_{x-1} \psi \partial_{\eta_x} \phi$. Hence, we get 
\begin{equation*}
\int d\nu (\eta)\,  (\eta_{x+1}\partial_{\eta_x} \phi- \eta_{x} \partial_{\eta_{x+1}} \phi)\,  \psi =0
\end{equation*}

This equation implies that $\nu (\eta_x, \eta_{x+1} | (e_y; y \neq x,x+1))$ is exchangeable. 

Let now $\Phi$ be a local test function of the form 
$$\Phi (\eta)= \prod_{s \in \ZZ} \phi_s (\eta_s)$$
where $(\phi_s)_s$ is a sequence of bounded smooth functions equal to $1$ for $|s| \ge A$ for a positive constant $A$. Our aim is to prove that for any $x$ we have
\begin{equation}
\label{eq:exc}
\nu (\Phi (\eta^{x,x+1})) = \nu (\Phi (\eta))
\end{equation}
which implies the exchangeability of $\nu$. We can assume that each $\phi_s$ is even or odd since every function can be decomposed as the sum of an even and an odd function. Moreover each even function $\phi_s (\eta_s)$ takes the form ${\tilde \phi}_s (e_s)$ for a suitable function ${\tilde \phi}_s$. 
 
If one of the $\phi_s$ is odd,  since $\nu$ is invariant by all flip operators, (\ref{eq:exc}) is trivial because the two terms are equal to zero. We assume that all the $\phi_s$ are even so that $\Phi$ is in fact a function depending only of the energies $e_s$ and we write $\Phi(\eta) ={\tilde \Phi} (e)=\prod_{s \in \ZZ} {\tilde \phi}_s (e_s)$. We shall denote by ${\tilde \nu}$ the law of $e:=\{e_y \, ; \, y \in \ZZ\}$. We have
\begin{equation*}
\begin{split}
&\int \Phi (\eta) d\nu(\eta) =\int d{\tilde \nu} (e) {\tilde \Phi} (e)\\
&=\int d{\tilde \nu} (e_y ; y \ne x,x+1) \left( \int {\tilde \Phi} (e)d{\tilde \nu} (e_x, {e_{x+1}} | e_y, y \ne x,x+1)\right)\\
&=\int d{\tilde \nu} (e_y; y \ne x,x+1) \left( \int {\tilde \Phi} (e^{x,x+1})d{\tilde \nu} (e_x, {e_{x+1}} | e_y, y \ne x,x+1)\right)\\
&=\int \Phi (\eta^{x,x+1}) d\nu(\eta)
\end{split}
\end{equation*}  
where we used the exchangeability of $\nu (\eta_x, \eta_{x+1} | (e_y; y \neq x,x+1))$ in the third equality. It concludes the proof that $\nu$ is exchangeable.

Hence, we can express $\nu$ as a mixture of $\mu_\beta$, $\beta \in (0,+\infty]$, with the convention that $\mu_{\infty}$ is the Dirac mass concentrated on the configuration $\delta_0$:
\begin{equation*}
\nu = \int_{(0,+\infty]} d\lambda (\beta) \mu_\beta
\end{equation*} 
In fact, $\lambda$ is the law under $\nu$ of the random variable $1/ {\mc E}(\eta)$. It remains to prove that $\nu ({\mc E}(\eta)=0)=\lambda (\{ +\infty\})=0$. It is a simple consequence of the fact that $H_{\Lambda_k} (\nu | {\mu}_{\bar \beta}) \le C_0 |\Lambda_k|$ for any $k$ and in particular for $k=0$. By (\ref{eq:vfh0}), we have that for any positive real $M$
\begin{equation*}
C_0 \ge M \int {\bf 1}_{\{0\} }(x) d\nu {\Big|}_{\{0\}} (x) -\log \left( \int e^{M {\bf 1}_{\{0\}} (x) } g_{\bar \beta} (x) dx \right)=M \lambda (\{ + \infty \})
\end{equation*} 
Since $M$ is arbitrary large, it follows that $\lambda(\{+\infty\})=0$.
\end{proof}

\section{Moments bounds}
\label{sec:mom-bounds}

The aim of this section is to give the proof of the following lemma:

\begin{Lemma}
\label{lem:mom-bounds}
Let $\mu^N$ be the probability measure $\mu_{\beta_0 (\cdot)}^N$ associated to a temperature profile bounded below by a strictly positive constant such that (\ref{eq:supen}) and (\ref{eq:entinit}) are valid. Let $(M_N)_{N\ge 1}$ be a sequence of positive integers such that $\liminf_{N \to \infty} M_N /N >0$. Then, there exists a positive constant $C$, which is independent of $N$, such that
\begin{equation*}
\sup_{t \ge 0} {\mathbb E}_{\mu^N} \left[ \cfrac{1}{M_N} \sum_{|x| \le M_N} e_{x} (t) \right] \le C
\end{equation*}
and
\begin{equation*}
\lim_{N \to \infty} \sup_{t \ge 0} {\mathbb E}_{\mu^N} \left[ \cfrac{1}{M_N^2 } \sum_{|x|\le M_N } e_x^2 (t) \right] =0
\end{equation*} 
\end{Lemma}

Let us first explain why the second equality of this lemma is nontrivial. The standard arguments to get moment upper bounds are based on the entropy inequality (\ref{eq:enti}) and the existence of exponential moments. In our case it would be necessary to have $\mu_\beta (e^{\alpha \eta_0^4}) < +\infty$ for $\alpha$ sufficiently  small. This is false since $\mu_\beta$ is a Gaussian measure. In \cite{B}, following an idea of Varadhan, and despite the absence of exponential moments, the use of the entropy inequality for the {\textit{microcanonical}} measure was sufficient to get a weak form of the lemma we want to prove. This approach cannot be carried here because we are in infinite volume and because the Dirichlet form is too degenerate to reproduce the argument.

\begin{proof}

The first statement is a simple consequence of the entropy inequality (\ref{eq:enti}). Indeed, for any $\delta>0$, we have
\begin{eqnarray*}
{\mathbb E}_{\mu^N}  \left[\cfrac{1}{M_N} \sum_{|x| \le M_N} e_t (x) \right]&\le& \cfrac{H(\mu_{t}^N | \mu_{\bar \beta})}{\delta M_N} +\cfrac{1}{\delta M_N }\log \left( \int  e^{\delta \sum_{|x| \le M_N} \eta_x^2 /2 } d\mu_{\bar \beta} (\eta) \right)
\end{eqnarray*}
The first term on the right-hand side is of order one by (\ref{eq:enttime}) and the second term is also of order one if $\delta$ is sufficiently small. Hence the left-hand side is of order one in $N$ uniformly in time.

The bound on the second moment of the energy is more difficult to obtain and the entropy inequality is not sufficient. We exploit here the Gaussian structure of the initial state. 

Recall the integral equations (\ref{eq:dyneq}) defining the dynamics. Each Poisson process ${\mc N}_x$ is interpreted as a clock and a jump of ${\mc N}_x$ as a ring of the clock. Conditionally to the realization of ${\mc N}= ({\mc N}_x)_x$, the dynamics is linear, thus the law remains Gaussian in the time interval between two successive rings. When a clock rings the flip operation conserves the Gaussian property of the state. Hence, conditionally to ${\mc N}$, the state remains Gaussian for any time. It follows that  the law $\mu^N_t $ of the process at time $t$ is a convex combination of Gaussian measures $G_{m,C}$ with mean $m \in \RR^{\ZZ}$ and correlation matrix $C \in {\mathcal S}_{\ZZ} (\RR)$, the space of symmetric matrices indexed by $\ZZ$: 
\begin{equation*}
\mu^N_t = \int d\rho_t (m,C) G_{m,C} 
\end{equation*}
Moreover, the convex combination $\rho_t (m,C)$ is the law at time $t$ of the Markov process $(m(t),C(t))$ with formal generator $N^2 {\mc G}$ where
\begin{equation*}
\begin{split}
({\mathcal G} F)(m,C)&=\sum_{x,y} (C_{x+1,y} - C_{x-1,y} +C_{x,y+1}-C_{x,y}) \partial_{C_{x,y}}F\\
&+ \sum_x (m_{x+1} -m_{x}) \partial_{m_x} F +\sum_{x} [F(C^x,m^x) -F(C,m)] 
\end{split}
\end{equation*}
with $C^x$ given by
\begin{equation*}
(C^x)_{u,v} = 
\begin{cases}
C_{u,v} \text{  if  } [u \ne x \text{ and } v \ne x] \text{  or  } [u=v=x],\\
-C_{u,v} \text{  otherwise}
\end{cases}
\end{equation*}
and 
\begin{equation*}
(m^x)_u = (-1)^{\delta_0 (x-u)} m_u
\end{equation*}

In other words,  $(C(\cdot),m(\cdot))$ are the solutions of the following integral equations
\begin{equation*}
\begin{cases}
C_{x,y} (t')= (-1)^{{\mc N}_{x} (t') +{\mc N}_{y} (t')} \left( C_{x,y} (0) \right.\\
\left.\phantom{C_{x,y} (t')}-\int_{0}^{t'} (-1)^{{\mc N}_{x} (t') +{\mc N}_{y} (t')} \left[ C_{x+1 ,y} (s) -C_{x-1 ,y} (s) +C_{x,y+1} (s) -C_{x,y-1} (s) \right] ds\right)\\
m_x (t') =  (-1)^{{\mc N}_{x} (t')} \left( m_{x} (0) -\int_{0}^{t'} (-1)^{{\mc N}_{x} (t')} \left[ m_{x+1} (s) -m_{x-1} (s)\right] ds\right)
\end{cases}
\end{equation*}
with initial conditions
\begin{equation*}
m_x (0) =0, \quad C_{x,y} (0) =\delta_{0} (x-y) \beta_0^{-1} (x/N)
\end{equation*}
and $t'=tN^2$.

The existence and uniqueness of solutions is easily established (by the same methods as presented in section \ref{sec:model}) in the space $\aleph=\aleph_0 \times \aleph_1$, where
\begin{equation*}
\begin{split}
\aleph_{0}=\bigcap_{\alpha>0}  \left\{ m\in \RR^{\ZZ} \; ; \; \sum_{x} e^{-\alpha |x|} m_x^2  <+\infty \right\}\\
\aleph_1=\bigcap_{\alpha >0}  \left\{ C \in  {\mc S}_{\ZZ} (\RR)\; ; \; \sum_{x,y}  e^{-\alpha (|x| + |y|)} C_{x,y}^2 <+\infty \right\}
\end{split}
\end{equation*}
Observe that the initial condition belongs to $\aleph$. Moreover, for any $(m,C) \in \aleph$, the Gaussian measure with mean $m$ and correlation matrix $C$ is meaningful (see e.g. chapter 2 of \cite{DaPrato}). 

This Markov process conserves the three quantities
\begin{equation}
\label{eq:cq}
\sum_{x \in \ZZ} m_x^2, \quad \sum_{x,y \in \ZZ} C_{x,y}^2, \quad \sum_{x \in \ZZ} C_{x,x}
\end{equation}

The initial condition $\mu^N$ is such that $\rho_0$ is the Dirac mass concentrated on 
$$m=0, \quad C_{x,y} =\delta_{0} (x-y) \beta_0^{-1} (x/N)$$
Therefore, we have $m(t)=0$ for any $t \ge 0$. By denoting, with abuse of notations, by $\rho_t (C)$ the law of $C (t)$ at time $t$, we have by the two last conservation laws (\ref{eq:cq}) that
\begin{equation*}
\int d\rho_{t} (C) \left(\cfrac{1}{M_N^2} \sum_{x,y \in \ZZ^2} (C_{x,y} -{\bar \beta}^{-1} \delta_0 (x-y))^{2}\right) = \cfrac{1}{M_N^2} \sum_{x\in \ZZ} [\beta^{-1}_0 (x/N) -{\bar \beta}^{-1}]^2
\end{equation*}
Moreover, we have
\begin{eqnarray*}
{\mathbb E}_{\mu^N} \left[ \cfrac{4}{M_N^2 } \sum_{|x|\le M_N } e_x^2 (t) \right] = {M_N}^{-2} \sum_{|x| \le M_N} \int d\rho_t (C) G_{0,C} (\eta_x^4)\\
=\cfrac{3}{{M_N}^{2}}  \sum_{|x| \le M_N} \int d\rho_t (C)  C_{x,x}^2 \\
=3  \int d\rho_t (C) \left\{ \cfrac{1}{M_N^2} \sum_{|x| \le M_N}(C_{x,x}-{\bar \beta}^{-1})^2 +\cfrac{2}{{\bar \beta} M_N^2} \sum_{|x| \le M_N} C_{x,x}\right \} +O\left(\frac{1}{M_N}\right)
\end{eqnarray*}
where we used the fact that, for a Gaussian centered variable, the fourth moment is given by three times the square of the second one. 

Observe that
\begin{equation*}
\int d\rho_t (C) \left\{\cfrac{1}{M_N^2} \sum_{|x| \le M_N} C_{x,x}\right \} =2 {\mathbb E}_{\mu^N} \left[ \cfrac{1}{M_N^2} \sum_{|x| \le M_N} e_x (t) \right]
\end{equation*}
and this term is order $M_N^{-1}$ by the first part of the lemma.

Up to terms of order $M_N^{-1}$, we are left with 
\begin{eqnarray*}
 \int d\rho_t (C) \left\{ \cfrac{1}{M_N^2} \sum_{|x| \le M_N}(C_{x,x}-{\bar \beta}^{-1})^2 \right \}\\
 \le  \int d\rho_t (C) \left\{ \cfrac{1}{M_N^2} \sum_{x,y \in \ZZ^2}(C_{x,y}-{\bar \beta}^{-1} \delta_{0} (x-y) )^2 \right \}\\
 = \int d\rho_0 (C) \left\{ \cfrac{1}{M_N^2} \sum_{x,y \in \ZZ^2}(C_{x,y}-{\bar \beta}^{-1} \delta_{0} (x-y) )^2 \right\}\\
 =  \cfrac{1}{M_N^2} \sum_{x\in \ZZ}(\beta_{0}^{-1} (x/N) -{\bar \beta}^{-1})^2
 \end{eqnarray*}
since the penultimate sum is conserved by $(C (t))_{t \ge 0}$.  By the assumption (\ref{eq:supen}), the last term goes to zero as $N$ goes to infinity.
  \end{proof}

\section{Green-Kubo formula}
\label{sec:GK}
In this section we study the homogenization properties for the diffusion coefficient in the linear response theory framework. To present the results we have to introduce some notations. 

Let $(\gamma_x)_{x \in {\ZZ}}$ be a sequence of i.i.d. positive random variables satisfying the assumption
$${\mathbf P} \left[ \gamma_- \le \gamma_x \le \gamma_+ \right] =1$$
where ${\mathbf P}$ is the probability measure on ${\mathbb R}^{\ZZ}$ given by the law of  the disorder $\gamma=(\gamma_x)_{x \in \ZZ}$. The corresponding expectation is denoted by $\mathbf E$.

In this section, time is not accelerated by a factor $N^2$. We first consider the closed system of length $N \ge 1$ with periodic boundary conditions. Let $\TT_N=\{0,\ldots,N-1\}$ be the usual discrete torus of length $N$. The generator ${\mathcal L}_N$ of the system is given by (\ref{eq:generator}) with the sums over $x \in \ZZ$ replaced by $x \in \TT_N$. 

Linear response theory predicts that the diffusion coefficient $D:=D(\{\gamma\},\beta)$ appearing in (\ref{eq:chaleur}) is given by
\begin{equation}
\label{eq:GK1}
D = \lim_{\lambda>0, \lambda \to 0} \lim_{ N \to \infty} L_N (\lambda)
\end{equation}
where $L_N:=L^{\gamma,\beta}_N$ is the Laplace transform of the current-current correlation function. It is defined for $z \in H^+$, $H^+=\{ z \in \CC\, ; \, {\mf R} (z) >0\}$, by
\begin{equation*}
L_N (z)= \cfrac{{\beta^2}}{2 N} \int_0^ \infty dt e^{-z t}\Big \langle \sum_{x \in \TT_N} j_{x,x+1} (t), \sum_{y \in \TT_N} j_{y,y+1} (0) \Big \rangle
\end{equation*}
Here, $\langle \cdot , \cdot \rangle := \langle \cdot , \cdot \rangle_{\beta} $ denotes the scalar product in ${\mathbb L}^2 (\mu_{\beta}^N)$ where $$\mu_{\beta}^N(d\eta) = \prod_{x \in \TT_N} g_{\beta} (\eta_x) d\eta_x$$ is the  Gibbs equilibrium measure with inverse temperature $\beta>0$ on $\RR^{\TT_N}$. We also use the short notation $\langle \cdot \rangle_{\beta}:=\langle \cdot \rangle$ for the expectation with respect to $\mu_{\beta}^N$.

The  Laplace transform $L_N$ can be written as
\begin{equation*}
L_N (z)= \cfrac{\beta^2}{2N}\left \langle \sum_{x\in \TT_N} j_{x,x+1}, (z -{\mathcal L}_N)^{-1} \left( \sum_{y \in \TT_N} j_{y,y+1} \right) \right \rangle
\end{equation*}

Observe that the definition (\ref{eq:GK1}) is only formal since it is not clear a priori that the limits exist. 

We also consider the {\textit{homogenized}} Green-Kubo formula for the infinite volume dynamics. It is defined by
\begin{equation}
\label{eq:GK2}
{\bar D}(\beta) = \lim_{\lambda>0, \lambda \to 0}  L^{\beta} (\lambda)
\end{equation}
where $L:=L^{\beta}$ is the Laplace transform of the averaged current-current correlation function. It is defined for $z \in H^+$ by
\begin{equation*}
L (z)= \cfrac{\beta^2}{2} \int_0^ \infty dt e^{-z t} \ll j_{0,1} (t), j_{0,1} (0) \gg
\end{equation*}
where $\ll \cdot , \cdot \gg= \ll \cdot , \cdot \gg_{\beta}$ is the inner product defined for bounded local functions $f$ and $g$ by
\begin{equation*}
\ll f, g \gg_{\beta} ={\mathbf E} \left( \sum_{x \in \ZZ} \left[ \langle \tau_x f , g \rangle_{\beta} -\langle f \rangle_{\beta} \langle g \rangle_{\beta} \right] \right)
\end{equation*} 
We shall denote by ${\mathbb L}^2 (\ll \cdot \gg)$ the Hilbert space generated by the set of bounded local functions and the inner product $\ll \cdot, \cdot \gg$.

The aim of this section is to show the following homogenization result 

\begin{theo}
\label{th:homo}
For almost every realization of the disorder $\gamma$, the Green-Kubo formulas (\ref{eq:GK1}) and (\ref{eq:GK2}) converge and are equal: $D(\{\gamma\},\beta) = {\bar D}(\beta)$. Moreover, ${\bar D}$ is independent of $\beta$.
\end{theo}

We recall that the functions $L_N$ and $L$ are analytical functions on $H_+$ (see e.g. \cite{RS}, Theorem VIII.2).

\begin{Lemma}
There exists a constant $C:=C(\beta, \gamma_+)$, independent of $N$, $\gamma$ and $z\in H_+$, such that
$$|L_N (z)| \le C$$
\end{Lemma}

\begin{proof}
The proof is a simple consequence of Proposition 6.1 in \cite{KL} and of the fact that ${\mathcal S} j_{x,x+1} = -2 (\gamma_x +\gamma_{x+1}) j_{x,x+1}$ (see also Theorem 2 in \cite{B2}).
\end{proof}

Let $h^{N}_z:=h^{N}_z (\eta; \beta, \gamma)$ be the solution of the resolvent equation in $\LL^2 (\langle \cdot \rangle)$:
$$(z -{\mathcal L}_N) h^N_{z} = \sum_{x \in \TT_N} j_{x,x+1}$$
We have
\begin{equation*}
L_N (z) =\cfrac{\beta^2}{2} \left\langle h^N_z , \cfrac{1}{N} \sum_{y \in \TT_N} j_{y,y+1} \right \rangle
\end{equation*}
 
Let ${h}_z:=h_z (\eta; \beta)$ be the solution of the resolvent equation in $\LL^2 (\ll \cdot \gg)$:
$$(z -{\mathcal L}) h_{z} =  j_{0,1}$$
We have
\begin{equation*}
L (z) = \cfrac{\beta^2}{2}\ll h_z , j_{0,1} \gg
\end{equation*}

Observe that if $\eta$ is distributed according to $\mu_\beta$ then $\beta^{1/2} \eta$ is distributed according to $\mu_1$. Since  $h_z (\eta;1)=h_z (\eta;\beta)$ and $j_{x,x+1}$ is an homogeneous function of degree two in $\eta$, it follows that $L^{\beta} (z)=L^1 (z)$. This implies the independence of the diffusion coefficient with respect to $\beta$.

In the following lemma we give an explicit formula for $L(z)$ if ${\mf R} (z)$ is sufficiently large. 

We shall denote by ${\mathbb P}_{R.W.}$  the law of the two-dimensional simple symmetric random walk $(S_j)_{j\ge 0} =(S_j^1, S_j^2)_{j \ge 0}$ starting from $(0,1)$ and by ${\mathbb E}_{R.W.}$ the corresponding expectation. Let ${\tilde {\mathbb E}}$ be the annealed expectation ${\mathbf E} {\mathbb E}_{R.W.}$.

For any path $\{S_j\}_{\{j=0,\ldots,k\}}$ of length $k$, we define ${\varepsilon} (\{S\}_k)=  \prod_{j=0}^{k-1} ((S_{j+1} -S_j)\cdot \bw) \in \{\pm 1\}$, where $\bw$ is the vector $(1,1)$ and $\bx \cdot \by$ denotes the usual scalar product of the two vectors $\bx$ and $\by$ of $\RR^2$. We also introduce the random potential 
\begin{equation*}
\exp( -V_{z} (x,y))= \cfrac{1}{z + {\bf 1}_{x \ne y} (\gamma_x + \gamma_y)}
\end{equation*}

 \begin{Lemma}
 There exists $\lambda_0 >0$ such that, for any $z \in H_+$ with ${\mf R} (z) \ge \lambda_0$, the Laplace transform $L(z)$ is given by
 \begin{equation}
 \label{eq:Laplaceform}
 L(z) =-\cfrac{1}{2} \sum_{k=0}^{\infty} (-4)^{k} {\tilde {\mathbb E}} \left[ {\varepsilon} (\{S\}_k)  \, e^{-\sum_{j=0}^{k} V_{z} ( S_j)}\delta_{\pm 1} (S_k^2 -S_k^1)\right]
 \end{equation}
 \end{Lemma}
 
 \begin{proof}
Since the generator ${\mathcal L}$ maps a polynomial function to a polynomial function of the same degree, the solution of the resolvent equation is expected to be of the form  
\begin{equation*}
h_{z} (\eta) = \sum_{x,y \in \ZZ^2} \phi_{z} (x,y) \eta_x \eta_y
\end{equation*}
where $\phi_z (x,y)$, $(x,y) \in \ZZ^2$, is the (symmetric) solution of
\begin{equation}
\label{eq:phi0}
(z + (\gamma_x +\gamma_y){\bf 1}_{x \neq y}) \phi_{z} (x,y) +({\tilde \nabla} \phi_z) (x,y) =-{\cfrac{\delta_1 (x)\delta_0 (y) +\delta_{0} (x)\delta_1 (y)}{2}}
\end{equation}
with, for any function $u:\ZZ^2 \to \RR$,
$$({\tilde \nabla} u)(x,y)=(u(x,y+1) -u(x,y-1)) +(u(x+1,y)-u(x-1,y))$$

We shall denote by $\lambda$ the real part of $z \in H_+$. In the sequel we show that, if $\lambda$ is sufficiently large, a solution to (\ref{eq:phi0}) exists, so that $h_z$ is of the form given above. In fact, it is not difficult to show that a solution to (\ref{eq:phi0}) exists for every $z \in H_+$. 


The Laplace transform $L(z)$ is equal to
\begin{eqnarray*}
L(z)&=& \cfrac{\beta^2}{2} \ll h_z ,j_{0,1} \gg = \cfrac{1}{2} {\ll} { \left\{\sum_{x,y \in \ZZ^2}  \phi_z (x,y)  \eta_x \eta_y\right\} , j_{0,1}} \gg \\
&=& -\cfrac{\beta^2}{2}\sum_{x,y} {\mathbf E} \left[ \phi_{z} (x,y) \lim_{n \to \infty} \cfrac{1}{2n+1} \sum_{|k| \le n} \langle \eta_x \eta_y \eta_k \eta_{k+1}\rangle \right]\\
&=&\cfrac{1}{2} \sum_{x,y} {\mathbf E} \left[ \phi_{z} (x,y) (\delta_{1}(x-y) + \delta_{-1} (x-y)\right]\\
&=&{\mathbf E} \left[\sum_{x \in \ZZ} \phi_z (x,x+1) \right]
\end{eqnarray*}

We define the operator $T_z$, acting on the set of real valued functions $u$ on ${\ZZ^2}$, by
\begin{equation}
\label{eq:T1}
(T_z u)(x,y) =\cfrac{1}{(z + (\gamma_x +\gamma_y){\bf 1}_{x \neq y})}({\tilde \nabla} u) (x,y)
\end{equation}
Then (\ref{eq:phi0}) can be written in the following form
\begin{equation*}
\phi_z + T_z\phi_z = \rho_z
\end{equation*}
where $\rho_z$ is the function given by
\begin{equation*}
\rho_z (x,y) = -\cfrac{(\delta_1 (x) \delta_0 (y) +\delta_{0} (x)\delta_1 (y))}{2 (z + (\gamma_x +\gamma_y){\bf 1}_{x \neq y})}
\end{equation*}

Observe that $\|T_z\phi \|_{\infty} \le (4/\lambda) \| \phi \|_{\infty}$ so that if $\lambda>4$ then $T_z$ is contractive for the $\| \cdot \|_{\infty}$ norm. It follows that for $\lambda$ sufficiently large

\begin{equation*}
\phi_{z} = \sum_{k=0}^{\infty} (-1)^k T_z^k \rho_z
\end{equation*}

For any $\bx \in \ZZ^2$, we have the following representation of the operator $T_z^k$
\begin{equation}
\label{eq:T2}
\begin{split}
&\; \;\; (T_z^k u)(\bx)\\
&= \sum_{|e_1|=1}\ldots\sum_{|e_k|=1} (e_1 \cdot \bw) \ldots (e_k\cdot \bw) e^{-\sum_{j=0}^{k-1}V_{z} (\bx+e_1+\ldots+e_j) } u(\bx +e_1+\ldots+e_k) 
\end{split}
\end{equation}
with the convention that the term in the exponential corresponding to $j=0$ is $V_z (\bx)$.
We obtain
\begin{equation*}
\begin{split}
&\phi_z (x,x+1)\\
&=-\cfrac{1}{2} {\mathbb E}_{R.W.}  \left[ \sum_{k=0}^{\infty}(-4)^k  {\varepsilon (\{S\}_k) }e^{-\sum_{j=0}^k V_{z} (S_j +(x,x))}\times  \right.\\
&\left. {\phantom{ \sum_{k=0}^{\infty}(-4)^k}} \left[ \delta_1 (x+S_k^1) \delta_0 (x+S_k^2) + \delta_ 0 (x+S_k^1) \delta_1 (x+S_k^2) \right]\right]
\end{split}
\end{equation*}
By summing over $x \in \ZZ$ and by taking the expectation with respect to the disorder, we obtain
\begin{eqnarray*}
L(z)&=&-\cfrac{1}{2} \sum_{k=0}^{\infty} (-4)^{k} {\tilde {\mathbb E}}  \left[  {\varepsilon (\{S\}_k) }e^{-\sum_{j=0}^k V_{z} (S_j -(S_{k}^2,S_{k}^2)) } {\bf 1}_{S_k^2 +1=S_k^1} \right]\\
&-&\cfrac{1}{2} \sum_{k=0}^{\infty} (-4)^{k} {\tilde {\mathbb E}}  \left[ {\varepsilon (\{S\}_k) }e^{-\sum_{j=0}^k V_{z} (S_j -(S_{k}^1,S_{k}^1))}{\bf 1}_{S_k^2 -1=S_k^1}\right]
\end{eqnarray*}
By taking first the expectation with respect to $\gamma$, we see that we can translate the environment and hence the potential by $(S_k^2, S_k^2)$ in the first expectation and by $(S_k^1, S_k^1)$ in the second one. Therefore, we get (\ref{eq:Laplaceform}).

 \end{proof}
 
\begin{Lemma}
There exists $\lambda_0 >0$ such that, for any $z \in H_+$ with ${\mf R} (z) \ge \lambda_0$ and almost every disorder $\gamma$, the limit of $L_N (z)$ as $N$ goes to infinity exists and is given by
\begin{equation}
-\cfrac{1}{2}\sum_{k=0}^{\infty} (-4)^{k} {\tilde {\mathbb E}} \left[ {\varepsilon} (\{S\}_k)  \, e^{-\sum_{j=1}^{k} V_{z} ( S_j)}\delta_{\pm 1} (S_k^2 -S_k^1)\right]
\end{equation}
\end{Lemma} 

\begin{proof} 
The proof is very similar to the previous one. We look for a solution in the form

\begin{equation*}
h^N_{z} (\eta) = \sum_{x,y} \phi_{z} (x,y) \eta_x \eta_y
\end{equation*}
with $\phi_z (x,y)$, $(x,y) \in {\TT}^2_N$, the solution of
\begin{equation}
\label{eq:phi}
(z + (\gamma_x +\gamma_y){\bf 1}_{x \neq y}) \phi_{z} (x,y) +({\tilde \nabla} \phi_z) (x,y) =-{\cfrac{\delta_1 (x-y) +\delta_{-1} (x-y)}{2}}
\end{equation}


Let $\lambda$ be the real part of $z \in H_+$ and define the operator $T_z$, acting on the real valued functions on ${\TT_N^2}$, according to (\ref{eq:T1}). Then (\ref{eq:phi}) can be written in the form $\phi_z + T_z \phi_z = \rho_z$, where $\rho_z$ is the function given by
\begin{equation*}
\rho_z (x,y) =- \cfrac{(\delta_1 (x-y) +\delta_{-1} (x-y))}{2 (z+ (\gamma_x +\gamma_y){\bf 1}_{x \neq y})}
\end{equation*}

Observe that $\|T_z u \|_{\infty} \le (4/\lambda) \| u \|_{\infty}$ so that if $\lambda>4$ then $T_z$ is contractive for the $\| \cdot \|_{\infty}$ norm. Therefore we have the following representation of $\phi_z$

\begin{equation*}
\phi_{z} = \sum_{k=0}^{\infty} (-1)^k T_z^k \rho_z 
\end{equation*}

For any $\bx \in \TT_N^2$ we have
\begin{eqnarray*}
&&(T_z^k u)(\bx)\\
&=& \sum_{|e_1|=1}\ldots\sum_{|e_k|=1} (e_1 \cdot \bw) \ldots (e_k\cdot \bw) e^{-\sum_{j=0}^{k-1}V_{z} (\bx+e_1+\ldots+e_j) } u(\bx +e_1+\ldots+e_k)
\end{eqnarray*}

Hence, we obtain
\begin{equation*}
(T_z^k \rho_z) (x,x+1)= -\cfrac{1}{2} (-4)^{k} {\mathbb E}_{R.W.} \left[ {\varepsilon} (\{S\}_k)  e^{-\sum_{j=1}^{k} V_{z} ((x,x) +S_j)} \delta_{\pm 1} (S_k^2 -S_k^1) \right]
\end{equation*}

Since $V_{z} ((x,x)+S_j) =\tau_x V_{z} (S_j )$, the ergodic theorem implies
\begin{equation*}
\lim_{N \to \infty} \cfrac{1}{N} \sum_{x \in \TT_N} \phi_{z} (x,x+1) = -\cfrac{1}{2}\sum_{k=0}^{\infty} (-4)^{k} {\tilde {\mathbb E}} \left[ \varepsilon (\{ S\}_k) \, e^{-\sum_{j=1}^{k} V_{z} (S_j)}\delta_{\pm 1} (S_k^2 -S_k^1)\right]
\end{equation*}

This completes the proof.


\end{proof}

Since the sequence $(L^{\gamma,\beta}_N (z))_N$ is a bounded sequence of analytical functions on $H_+$, Montel theorem implies it forms a compact sequence in the Banach space of analytical functions. Let $L_{\infty}^{\gamma^1, \beta}, L_{\infty}^{\gamma^2,\beta}$ be any (analytical) limit points corresponding to the realizations of $\gamma^1$ and $\gamma^2$ of the disorder. For any $z \in H_+$ such that  $\{{\mf R} (z) \ge \lambda_0 \}$, we have
\begin{equation*}
L_{\infty}^{\gamma_1,\beta} (z) =L^{\beta} (z) = L_{\infty}^{\gamma_2,\beta} (z)
\end{equation*}
Since the two analytical functions $L_{\infty}^{\gamma_1,\beta}$ and $L_{\infty}^{\gamma_2,\beta}$ on $H_+$ \, coincide \, on $\{z \, ; \, {\mf R} (z) \ge \lambda_0\}$ with $L^{\beta}$, they are equal on $H_+$ to $L^{\beta}$. It follows that, for almost every realization of the disorder and every $z \in H_+$, the limit as $N$ goes to infinity of $L^{\gamma,\beta}_N (z)$ exists and is equal to $L^{\beta} (z)$. The theorem is a trivial consequence of the following non trivial fact:

\begin{Lemma}
The limit, as $\lambda \in (0,+\infty)$ goes to $0$, of $L^{\beta} (\lambda)$ exists.
\end{Lemma}

\begin{proof}
The proof is similar to the proof of Theorem 1 in \cite{B2} (see also \cite{Bo}). 
\end{proof}

\end{document}